\documentclass{article}
\usepackage{authblk}
\usepackage[square]{natbib}
\usepackage{yhmath}
\usepackage{mathtools, amsmath, amsthm, amssymb, amsbsy, amsfonts, amscd, stmaryrd}
\usepackage{euscript}
\usepackage{bm}
\usepackage{xcolor, soul}
\usepackage{empheq}
\usepackage{rotating}
\usepackage{enumerate}
\usepackage{enumitem}
\usepackage{bbold}
\usepackage{ftnxtra}
\usepackage{fnpos}
\usepackage{appendix}
\usepackage{graphicx}
\graphicspath{{Figures/}}
\usepackage{pstool}
\usepackage{psfrag}
\usepackage{tikz-cd}
\usepackage{float}
\usepackage{comment}
\numberwithin{equation}{section}
\theoremstyle{plain}	
 \newtheorem{thm}{Theorem}[section]
 
 \newtheorem{lem}[thm]{Lemma}
 \newtheorem{prop}[thm]{Proposition}
\theoremstyle{definition}	
 
 \newtheorem{remark}{Remark}[section]
 
\usepackage{caption}
\usepackage{hyperref}
\hypersetup{colorlinks=true, linkcolor=blue}
\hypersetup{colorlinks=true,citecolor=blue}
\DeclareMathAlphabet{\mathpzc}{OT1}{pzc}{m}{it}
\usepackage{mathrsfs}
\definecolor{lighter_purple_mathematica}{rgb}{0.6666666666,0.33333333333,0.666666666666}
\makeatletter
\newsavebox{\@brx}
\newcommand{\llangle}[1][]{\savebox{\@brx}{$\m@th{#1\langle}$}%
 \mathopen{\copy\@brx\mkern2mu\kern-0.9\wd\@brx\usebox{\@brx}}}
\newcommand{\rrangle}[1][]{\savebox{\@brx}{$\m@th{#1\rangle}$}%
 \mathclose{\copy\@brx\mkern2mu\kern-0.9\wd\@brx\usebox{\@brx}}}%
\let\oldabs\abs
\def\abs{\@ifstar{\oldabs}{\oldabs*}}
\makeatother

\newcommand{\M}{\mathbf M}
\newcommand{\cM}{\mathrm M}
\newcommand{\m}{\mathbf m}

\newcommand{\Lt}{\mathbf L}
\newcommand{\Lo}{\mathring{\Lt}}
\newcommand{\cL}{\mathrm L}
\newcommand{\cLo}{\mathring{\cL}}
\newcommand{\Lv}{\accentset{v}{\Lt}}
\newcommand{\cLv}{\accentset{v}{\cL}}
\newcommand{\Li}{\accentset{i}{\Lt}}
\newcommand{\cLi}{\accentset{i}{\cL}}
\newcommand{\g}{\mathbf g}
\newcommand{\G}{\mathbf G}
\newcommand{\Go}{\mathring{\G}}
\newcommand{\Gv}{\accentset{v}{\G}}
\newcommand{\Gi}{\accentset{i}{\G}}
\newcommand{\F}{\mathbf{F}}
\newcommand{\C}{\mathbf{C}}
\newcommand{\fb}{\mathbf{b}}
\newcommand{\ic}{\mathbf{c}}
\newcommand{\cg}{\mathrm g}
\newcommand{\cG}{\mathrm G}
\newcommand{\cGo}{\mathring{\cG}}
\newcommand{\cF}{\mathrm F}
\newcommand{\cP}{\mathrm P}
\newcommand{\cC}{\mathrm{C}}
\usepackage{accents}
\newcommand{\Fe}{\accentset{e}{\F}}
\newcommand{\cFe}{\accentset{e}{\mathrm{F}}}
\newcommand{\Fa}{\accentset{a}{\F}}
\newcommand{\cFa}{\accentset{a}{\mathrm{F}}}
\newcommand{\Fv}{\accentset{v}{\F}}
\newcommand{\cFv}{\accentset{v}{\mathrm{F}}}
\newcommand{\Fve}{\accentset{ve}{\F}}
\newcommand{\cFve}{\accentset{ve}{\mathrm{F}}}
\newcommand{\Ce}{\accentset{e}{\C}}
\newcommand{\cCe}{\accentset{e}{\mathrm{C}}}
\newcommand{\Be}{\accentset{e}{\mathbf{B}}}
\newcommand{\cBe}{\accentset{e}{\mathrm{B}}}
\newcommand{\ce}{\accentset{e}{\mathbf{c}}}
\newcommand{\cce}{\accentset{e}{\mathrm{c}}}
\newcommand{\be}{\accentset{e}{\mathbf{b}}}
\newcommand{\cbe}{\accentset{e}{\mathrm{b}}}

\newcommand{\Lambdav}{\accentset{v}{\boldsymbol\Lambda}}
\newcommand{\Lambdai}{\accentset{i}{\boldsymbol\Lambda}}
\newcommand{\sharpo}{{\mathring{\sharp}}}
\newcommand{\flato}{\mathring{\flat}}

\usepackage[all,cmtip]{xy}
\begin{document}
\title{\textbf{On Physical Components of Tensors \\in Elasticity and Inelasticity
}}
\author[1]{Souhayl Sadik\thanks{E-mail: sosa@mpe.au.dk}}
\author[2,3]{Arash Yavari\thanks{E-mail: arash.yavari@ce.gatech.edu}}
\affil[1]{\small \textit{Department of Mechanical and Production Engineering, Aarhus University, 8000~Aarhus~C, Denmark}}
\affil[2]{\small \textit{School of Civil and Environmental Engineering, Georgia Institute of Technology, Atlanta, GA 30332, USA}}
\affil[3]{\small \textit{The George W. Woodruff School of Mechanical Engineering, Georgia Institute of Technology, Atlanta, GA 30332, USA}}
\maketitle
%-----------------------------
\begin{abstract}
Widely used in mechanics and mathematical physics, physical components remove the inherent coordinate-dependent scaling in curvilinear coordinates, yielding components of consistent physical dimension. In orthogonal coordinates, they are constructed by normalizing the coordinate frame and coframe; in general coordinates, however, their construction requires additional, non-trivial choices. In this paper, we extract physical components of arbitrary tensors on arbitrary Riemannian manifolds by orthonormalization of the coordinate frame. We further formulate a general normalization framework distinguishing three requirements: dimensional consistency, dual frame-coframe compatibility, and unit normalization. We show that dimensional consistency alone leaves independent general linear gauge freedoms for the contravariant and covariant components. Further requiring dual compatibility locks these into a single general linear gauge; and a dual-compatible frame and coframe are both of unit length if and only if they are orthonormal. Thus, the choice of orthonormal transformations emerges as the only physical components framework that satisfies all three requirements. We apply this framework to nonlinear elasticity and inelasticity—referring broadly to constitutive responses involving internal distortions, of which we study anelasticity, viscoelasticity, and visco-anelasticity—examining the deformation gradient, inelastic distortions, strain measures, and stress tensors. We conclude by arguing that physical components remain neither intrinsic nor unique.
\end{abstract}
%-----------------------------
\begin{description}
\item[Keywords:] Physical components of tensors, nonlinear elasticity, inelasticity, material metric, two-point tensors, non-orthogonal coordinates.
\end{description}
\tableofcontents
%-----------------------------
%-----------------------------
\section{Introduction}
\label{Section:Intro}
In curvilinear coordinates, the coordinate functions do not, in general, have the same physical dimension, and increments in different coordinate directions do not necessarily correspond to equal physical quantities. For example, in cylindrical coordinates $(r,\theta,z)$, the radial and axial coordinates $r$ and $z$ have the dimension of length, whereas the angular coordinate $\theta$ is dimensionless. Indeed, the actual displacement in the $\theta$-direction corresponds to a physical length $r d\theta$, which depends on the position. More generally, the coordinate basis vectors are not necessarily of unit length, and their norms are determined by the metric tensors. 
As a result, the coordinate components of a tensor reflect both its tensorial character and the scaling induced by the choice of coordinates.
Different components of a tensor may have different physical dimensions \citep{Truesdell1953physical,Aris1962}. 
Depending on the choice of coordinates in the reference and current configurations, some components of the deformation gradient may not be dimensionless.
In order to obtain components dimensionally consistent with the nature of the tensor (for example, dimensionless components for deformation gradient $\mathbf{F}$ in orthogonal coordinates), one expresses the tensor with respect to orthonormal bases relative to the corresponding metrics. This removes the inherent coordinate-dependent scaling in curvilinear coordinates.

The notion of physical components of tensors in curvilinear coordinates can be traced back to the work of \citet{RicciLeviCivita1901}. In their treatment of vector analysis in generalized coordinates, they distinguished between covariant and contravariant quantities and the corresponding projections and components measured with respect to unit directions associated with the coordinate lines and coordinate surfaces. Although the modern terminology of ``physical components'' had not yet emerged, their construction anticipates the contemporary interpretation of physical components as tensor components referred to orthonormal frames. \citet{Hencky1928} also employed physical components in curvilinear coordinates in his work on nonlinear elasticity. \citet{McConnell1931AbsoluteCalculus} provided one of the first systematic definitions in orthogonal curvilinear coordinates by identifying physical components with projections onto unit tangent vectors to coordinate lines. \citet{SyngeSchild1949}, as well as \citet{GreenZerna1954}, extended these ideas to tensors of higher order within continuum mechanics. \citet{Truesdell1953physical} further clarified the role of physical components and emphasized their relation to intrinsic tensorial structure. \citet{Aris1962} presented a formulation of physical components in orthogonal coordinate systems in the context of fluid mechanics. It should be noted that in non-orthogonal coordinate systems the definition of physical components is not unique, as noted by \citet{Ericksen1960}. A concise historical account and synthesis of these developments is given in \citep{Altman1982PhysicalComponents}.

This paper is organized as follows. In \S\ref{Section:Physical-Components}, we define the physical dimensions of tensor fields and develop the physical-component representation of tensors in orthogonal and general curvilinear coordinates. We also extend the construction to two-point tensors and discuss the non-uniqueness associated with the choice of orthonormal frames. A general normalization framework is introduced to distinguish dimensional consistency, dual compatibility, and unit normalization, and to characterize the orthonormal-frame construction in terms of these requirements. In \S\ref{Section:Anelasticity}, after briefly reviewing the geometric framework of nonlinear elasticity and anelasticity, we derive physical components of the deformation gradient, strain measures, and stress tensors. Particular attention is given to the distinct metrics involved in anelasticity and to the physical-component representations of the elastic and anelastic distortions. The framework is then extended for inelastic processes beyond anelasticity and worked out for viscoelasticity and visco-anelasticity.
In \S\ref{Section:Example}, an example of solids with distributed eigenstrains is considered, and physical components of the relevant tensor fields are explicitly calculated.
Some concluding remarks are given in \S\ref{Section:Conclusions}.
%-----------------------------
%-----------------------------
\section{Physical components of tensors}
\label{Section:Physical-Components}
In this section, we discuss the notion of physical components for tensor fields on Riemannian manifolds. Single-manifold constructions are carried out on a Riemannian manifold $(\mathcal{M},\M)$, while two-point tensors are treated on a pair of Riemannian manifolds $(\mathcal{M},\M)$ and $(\mathcal{N},\m)$ over a smooth map $\phi:\mathcal{M}\to\mathcal{N}$.

%-----------------------------
%-----------------------------
\subsection{Physical dimensions of tensor fields}
\label{S:Phys_D}
Let $\{u_1,\dots,u_k\}$ denote a set of fundamental units (e.g., international system of units, British imperial system, US customary system). A change of the system of units is defined by rescaling these fundamental units by positive factors $\lambda_1,\dots,\lambda_k \in \mathbb{R}^+$, respectively. If under this change of units a scalar physical quantity $\phi$ transforms as $\phi \mapsto \lambda_1^{\alpha_1}\cdots \lambda_k^{\alpha_k}\,\phi$, we write $[\phi]=[u_1]^{\alpha_1}\cdots [u_k]^{\alpha_k}$,
where $[\,\cdot\,]$ denotes the \textit{physical dimension} \citep{Bridgman1922,Barenblatt2003}. For example, if $\phi$ represents a mass density, then under rescaling of length and mass units it transforms as $\phi \mapsto \lambda_{\text{mass}}\,\lambda_{\text{length}}^{-3}\,\phi$, and hence $[\phi]=[\text{mass}]\,[\text{length}]^{-3}=M L^{-3}$. Thus, the physical dimension of $\phi$ is completely characterized by its scaling behavior under changes of the fundamental units.
Let $\mathbf{U}$ be a vector field. If under a change of the system of units it transforms as $\mathbf{U} \mapsto \lambda_1^{\beta_1}\cdots \lambda_k^{\beta_k}\,\mathbf{U}$, one writes $[\mathbf{U}]=[u_1]^{\beta_1}\cdots [u_k]^{\beta_k}$.
This definition is coordinate-free: the vector field itself is multiplied by a scalar under a change of units, while its components may have different physical dimensions depending on the choice of coordinates. For example, a velocity vector has physical dimension $LT^{-1}$.
Let $\boldsymbol{\alpha}$ be a $1$-form. If under a change of the system of units it transforms as $\boldsymbol{\alpha} \mapsto \lambda_1^{\gamma_1}\cdots \lambda_k^{\gamma_k}\,\boldsymbol{\alpha}$,
we write $[\boldsymbol{\alpha}]=[u_1]^{\gamma_1}\cdots [u_k]^{\gamma_k}$.
This definition is again coordinate-free. For example, consider a force field $\mathbf{F}$, viewed as a $1$-form. When acting on a velocity vector $\mathbf{V}$ of a particle, the scalar field $\langle \mathbf{F},\mathbf{V} \rangle$ represents power, and hence has physical dimension $M L^2 T^{-3}$.

A tensor of type $\binom{r}{s}$ at $X \in \mathcal{M}$ is a multilinear map
%-----------------------------
\begin{equation}
\label{eq:(r,s)_tensor}
 \mathbf{T}:\underbrace{T_X^{*}\mathcal{M}\times \cdots \times T_X^{*}\mathcal{M}}_{r~\mathrm{copies}}\times
 \underbrace{T_X\mathcal{M}\times \cdots \times T_X\mathcal{M}}_{s~\mathrm{copies}}
 \rightarrow \mathbb{R}\,,
\end{equation}
%-----------------------------
where $T_X \mathcal{M}$ is the tangent space of $\mathcal{M}$ at $X$, and $T_X^{*}\mathcal{M}$ is its dual \citep{MarsdenHughes1983}. The tensor $\mathbf{T}$ is said to be contravariant of order $r$ and covariant of order $s$. Let $\{X^A\}$ be a local coordinate chart on $\mathcal{M}$. $\mathbf{T}$ admits the following coordinate representation
%-----------------------------
\begin{equation}
	\mathbf{T}=T^{A_1\dots A_r}{}_{B_1\dots B_s}\,\frac{\partial}{\partial X^{A_1}}\otimes\cdots\otimes
	\frac{\partial}{\partial X^{A_r}}\otimes dX^{B_1}\otimes\cdots\otimes dX^{B_s}\,,
\end{equation}
%-----------------------------
with $T^{A_1\dots A_r}{}_{B_1\dots B_s}$ denoting the components of $\mathbf{T}$ with respect to the coordinate basis $\left\{{\partial_A}\right\}$ and its dual basis $\{dX^A\}$, such that, for $\boldsymbol{\alpha}^k=\alpha^k_A\,dX^A\in T_X^*\mathcal{M}$ and $\mathbf{V}_k=V_k^A\,{\partial_A}\in T_X\mathcal{M}$, one may write
%-----------------------------
\begin{equation}
 \mathbf{T}(\boldsymbol{\alpha}^1,\dots,\boldsymbol{\alpha}^r,\mathbf{V}_1,\dots,\mathbf{V}_s)
 = T^{A_1\dots A_r}{}_{B_1\dots B_s}\,\alpha^1_{A_1}\cdots \alpha^r_{A_r} V_1^{B_1}\cdots V_s^{B_s}\,.
\end{equation}
%-----------------------------

A two-point tensor of type $\binom{r~r'}{s~s'}$ over a map $\phi:\mathcal{M}\rightarrow \mathcal{N}$ is a multilinear map
%-----------------------------
\begin{equation}
\begin{aligned}
 \mathbf{T}:\;\underbrace{T_X^{*}\mathcal{M}\times \cdots \times 
 T_X^{*}\mathcal{M}}_{r~\mathrm{copies}}\times
 \underbrace{T_X\mathcal{M}\times \cdots \times T_X\mathcal{M}}_{s~\mathrm{copies}} 
 \times
 \underbrace{T_x^{*}\mathcal{N}\times \cdots \times 
 T_x^{*}\mathcal{N}}_{r'~\mathrm{copies}}\times
 \underbrace{T_x\mathcal{N}\times \cdots \times T_x\mathcal{N}}_{s'~\mathrm{copies}}
 \rightarrow \mathbb{R}\,,
\end{aligned}
\end{equation}
%-----------------------------
where $X\in \mathcal{M}$ and $x=\phi(X)$.
The two-point tensor $\mathbf{T}$ is said to be contravariant of order $(r,r')$ and covariant of order $(s,s')$. Let $\{X^A\}$ and $\{x^a\}$ be local coordinate charts on $\mathcal{M}$ and $\mathcal{N}$, respectively.
$\mathbf{T}$ admits the following coordinate representation
%-----------------------------
\begin{equation} \label{Tensor-General-TwoPoint}
\begin{aligned}
	\mathbf{T} =\mathrm T^{A_1\dots A_r\,a_1\dots a_{r'}}{}_{B_1\dots B_s\,b_1\dots b_{s'}}\,
	& \frac{\partial}{\partial X^{A_1}}\otimes\cdots\otimes\frac{\partial}{\partial X^{A_r}}
	\otimes\frac{\partial}{\partial x^{a_1}}\otimes\cdots\otimes\frac{\partial}{\partial x^{a_{r'}}}\\
	&\otimes dX^{B_1}\otimes\cdots\otimes dX^{B_s}\otimes dx^{b_1}\otimes\cdots\otimes dx^{b_{s'}}\,,
\end{aligned}
\end{equation}
%-----------------------------
with $\mathrm T^{A_1\dots A_r\,a_1\dots a_{r'}}{}_{B_1\dots B_s\,b_1\dots b_{s'}}$ denoting the components of $\mathbf{T}$ in the coordinate bases $\left\{{\partial_A}\right\}$, $\left\{\frac{\partial}{\partial x^a}\right\}$ and their dual bases $\left\{d X^A\right\}$, $\left\{d x^a\right\}$, such that, for $\boldsymbol{\Pi}^{(k)}=\Pi^{(k)}_A\,dX^A\in T_X^*\mathcal{M}$, $\mathbf{W}_{(k)}=\mathrm W_{(k)}^A\,{\partial_A}\in T_X\mathcal{M}$, $\boldsymbol{\pi}^{(k)}=\pi^{(k)}_a\,dx^a\in T_x^*\mathcal{N}$, and $\mathbf{w}_{(k)}=\mathrm w_{(k)}^a\,\frac{\partial}{\partial x^a}\in T_x\mathcal{N}$, one may write
%-----------------------------
\begin{equation}
\begin{aligned}
	&\mathbf{T}(\boldsymbol{\Pi}^{(1)},\dots,\boldsymbol{\Pi}^{(r)},\mathbf{W}_{(1)},\dots,\mathbf{W}_{(s)},
	\boldsymbol{\pi}^{(1)},\dots,\boldsymbol{\pi}^{(r')},\mathbf{w}_{(1)},\dots,\mathbf{w}_{(s')})\\
	&\quad =T^{A_1\dots A_r\,a_1\dots a_{r'}}{}_{B_1\dots B_s\,b_1\dots b_{s'}}\,
	\Pi^{(1)}_{A_1}\cdots \Pi^{(r)}_{A_r} \mathrm W_{(1)}^{B_1}\cdots \mathrm W_{(s)}^{B_s} \pi^{(1)}_{a_1}\cdots \pi^{(r')}_{a_{r'}}
	\mathrm w_{(1)}^{b_1}\cdots \mathrm w_{(s')}^{b_{s'}}\,.
\end{aligned}
\end{equation}
%-----------------------------

The physical dimension of the tensor (or two-point tensor) $\mathbf{T}$ is defined by its scaling under a change of the system of units. If $\mathbf{T} \mapsto \lambda_1^{\delta_1}\cdots \lambda_k^{\delta_k}\,\mathbf{T}$, we write, independently of coordinates, that $[\mathbf{T}]=[u_1]^{\delta_1}\cdots [u_k]^{\delta_k}$.
%-----------------------------
\begin{remark}
\label{rmrk:norm}
The norm of a tensor field is defined using a metric and represents the magnitude of the corresponding physical quantity. Therefore, it must have the same physical dimension as the tensor itself. For example, the norm of a velocity vector is the speed, and both have physical dimension $LT^{-1}$.
If the norm of a tensor $\mathbf{T}$ does not have the same physical dimension as $\mathbf{T}$, then the metric used to define the norm introduces an inconsistent physical scaling. In this case, the metric does not correctly represent the physical geometry of the quantity under consideration. Consequently, a physically meaningful metric must be such that the norm of any tensor has the same physical dimension as the tensor itself.
\end{remark}
%-----------------------------
The metric-induced norms above rely on a further piece of notation, introduced here for use throughout the paper. On a Riemannian manifold $(\mathcal{M}, \M)$ the metric $\M$ supplies a canonical isomorphism between the tangent and cotangent bundles. The \emph{flat} map $\flat: T\mathcal{M} \to T^*\mathcal{M}$, $\mathbf{V}^\flat \coloneq \M(\mathbf{V}, \cdot)$, lowers an index according to $\mathrm{V}^A \cM_{AB} = \mathrm{V}_B$. Its inverse, the \emph{sharp} map $\sharp: T^*\mathcal{M} \to T\mathcal{M}$, $\boldsymbol\vartheta^\sharp \coloneq \M^{-1}(\boldsymbol\vartheta, \cdot)$, raises an index through $\vartheta_A \cM^{-AB} = \vartheta^B$, where $\cM^{-AB}$ are the components of $\M^{-1}$. The two operations act on a general $\binom{r}{s}$ tensor by contracting, respectively, its contravariant slots with $\M$ and its covariant slots with $\M^{-1}$. Applied to the metric itself, sharpening returns the inverse metric: $\M^\sharp = \M^{-1}\M\M^{-1} = \M^{-1}$, and hence in components $\cM^{AB}=\cM^{-AB}$.
%---------------------------------
%---------------------------------
\subsection{Physical components in orthogonal coordinates}
\label{Subsection:orthogonal}

Let $(\mathcal M,\M)$ be a Riemannian manifold, $\{X^A\}$ be an orthogonal curvilinear coordinate chart, $\left\{\partial_A\coloneq{\partial}/{\partial X^A}\right\}$ the corresponding (orthogonal) coordinate basis, and $\{dX^A\}$ its dual basis.
The matrix representation of the metric is hence $\llbracket\cM_{AB}\rrbracket$ is diagonal.\footnote{We reserve square brackets $[\cdot]$ for physical dimensions and use double brackets $\llbracket\cdot\rrbracket$ for matrix representations.}
Let $\mathbf{U}$ be an arbitrary vector with physical dimension $[\mathbf{U}]$ and $\boldsymbol\vartheta$ be an arbitrary 1-form with physical dimension $[\boldsymbol\vartheta]$. By definition of the metric-induced inner product, the scalar ${\|\mathbf{U}\|^2_{\M}=\llangle \mathbf{U},\mathbf{U} \rrangle_{\M}}$ has physical dimension $[\mathbf{U}]^2$, and the scalar ${\|\boldsymbol\vartheta\|^2_{\M^\sharp}=\llangle \boldsymbol\vartheta,\boldsymbol\vartheta \rrangle_{\M^\sharp}}$ has physical dimension $[\boldsymbol\vartheta]^2$---see Remark~\ref{rmrk:norm}. One hence writes 
%---------------------------------
\begin{equation}
\label{eq:norm_dim}
[\mathbf{U}]^2 = [\llangle \mathbf{U},\mathbf{U} \rrangle_{\M}]\,,\quad
[\boldsymbol\vartheta]^2 = [\llangle \boldsymbol\vartheta, \boldsymbol\vartheta \rrangle_{\M^\sharp}]\,.
\end{equation}
%---------------------------------
However, the coordinate representations $\mathbf{U}=\mathrm{U}^A \partial_A$ and $\boldsymbol\vartheta=\vartheta_A dX^A$ yield components $\{\mathrm{U}^A\}$ and $\{\vartheta_A\}$ that do not, in general, share the same physical dimensions as $\mathbf{U}$ and $\boldsymbol\vartheta$, respectively. In what follows, we proceed to define and derive \emph{physical components} of an arbitrary tensor with respect to an orthogonal curvilinear coordinate system.

Let us normalize the coordinate basis $\partial_A$ as
%---------------------------------
\begin{equation} \label{eq:Normalized-basis}
	\hat{\mathbf{E}}_A
	\coloneq \frac{1}{\left\|{\partial_A}\right\|_{\M}}{\partial_A}
	= \frac{1}{\sqrt{\llangle {\partial_A},{\partial_A} \rrangle_{\M}}}
	\,{\partial_A}
	= \frac{1}{\sqrt{\cM_{AA}}}\,{\partial_A}\quad (\text{no summation on }A)\,.
\end{equation}
%---------------------------------
%where we use $\left\|{\partial_A}\right\|_{\M} = \sqrt{\llangle {\partial_A},{\partial_A} \rrangle_{\M}} = \sqrt{\cM_{AA}}\quad (\text{no summation on }A)$.
We next show that the normalized basis vectors ${\hat{\mathbf{E}}_A}$ are dimensionless.% in the sense of physical dimensions.
Specializing \eqref{eq:norm_dim} for the coordinate basis vectors ${\partial_A}$ yields
$\left[{\partial_A}\right]^2
= \left[\llangle {\partial_A},{\partial_A} \rrangle_{\M}\right]
= \left[\cM_{AA}\right]$ (no summation on $A$).
Consequently, for the normalized basis vector $\hat{\mathbf{E}}_A$, we conclude it is dimensionless
%---------------------------------
\begin{equation}
	\left[\hat{\mathbf{E}}_A\right]
	= \frac{1}{\left[\sqrt{\cM_{AA}}\right]}\,\left[{\partial_A}\right]
	= \frac{\sqrt{\left[\cM_{AA}\right]}}{\sqrt{\left[\cM_{AA}\right]}}
	= 1\quad (\text{no summation on }A)\,.
\end{equation}
%---------------------------------
Let $\{\hat{\boldsymbol\Omega}^A\}$ be the dual coframe of $\{\hat{\mathbf{E}}_A\}$.
For the orthogonal coordinates $\{X^A\}$, we have
%---------------------------------
\begin{equation} \label{eq:Normalized-cobasis}
	\hat{\boldsymbol\Omega}^A=\sqrt{\cM_{AA}}\,dX^A\quad (\text{no summation on }A)\,.
\end{equation}
%---------------------------------
Specializing \eqref{eq:norm_dim} for the coordinate dual basis $\{dX^A\}$ yields
$\left[dX^A\right]^2 = \left[\llangle dX^A, dX^A\rrangle_{\M^\sharp}\right]=\left[\cM^{AA}\right]=1/\left[\cM_{AA}\right]$ (no summation on $A$),\footnote{Note that $\cM^{AA}=1/\cM_{AA}$ (no summation on $A$) follows from the fact that the metric representation $\llbracket\cM_{AB}\rrbracket$ is diagonal.}
and one finds that $\{\hat{\boldsymbol\Omega}^A\}$ is dimensionless.

Let write a vector $\mathbf{U}$ and a co-vector $\boldsymbol\vartheta$ in components using \eqref{eq:Normalized-basis} and \eqref{eq:Normalized-cobasis}
%---------------------------------
\begin{equation}
	\mathbf{U}
	=\mathrm{U}^A\,{\partial_A}
	=\sum_A\mathrm{U}^A\sqrt{\cM_{AA}}\,\,\hat{\mathbf{E}}_A
	=\hat{\mathrm U}^A\,\hat{\mathbf{E}}_A\,,
	\quad
	\boldsymbol\vartheta
	=\vartheta_A\,dX^A
	=\sum_A\frac{\vartheta_A}{\sqrt{\cM_{AA}}}\,\hat{\boldsymbol\Omega}^A
	=\hat\vartheta_A\,\hat{\boldsymbol\Omega}^A\,,
\end{equation}
%---------------------------------
where the components $\hat{\mathrm U}^A$ and $\hat\vartheta_A$ obtained in this way carry the same physical dimensions as $\mathbf{U}$ and $\boldsymbol\vartheta$ (unlike the coordinate components $\mathrm{U}^A$ and $\vartheta_A$), respectively, since the frames $\{\hat{\mathbf{E}}_A\}$ and $\{\hat{\boldsymbol\Omega}^A\}$ are dimensionless. The quantities $\hat{\mathrm U}^A$ and $\hat\vartheta_A$ are called \emph{physical components} of $\mathbf{U}$ and $\boldsymbol\vartheta$, respectively, relative to the orthonormal frame~$\{\hat{\mathbf{E}}_A\}$ and coframe~$\{\hat{\boldsymbol\Omega}^A\}$, respectively. Explicitly
%---------------------------------
\begin{equation}\label{eq:Phys_Comp}
	\hat{\mathrm U}^A\coloneq\sqrt{\cM_{AA}}\,\mathrm{U}^A\,,\quad
	\hat\vartheta_A\coloneq\frac{1}{\sqrt{\cM_{AA}}}\,\vartheta_A
	\quad (\text{no summation on }A)\,.
\end{equation}
%---------------------------------
%---------------------------------
\subsection{Physical components in general coordinates}
\label{Subsection:general}
In this section, we consider a general local curvilinear coordinate system $\{X^A\}$ on $\mathcal{M}$, i.e. one that is not necessarily $\M$-orthogonal.
The coordinate basis $\{\partial/\partial X^A\}$ is hence, in general, not orthogonal.\footnote{Note that the metric representation $\llbracket\cM_{AB}\rrbracket$ is hence not necessarily diagonal.}
In order to define physical components, we orthonormalize the frame and use the results of the previous subsection.
We proceed by constructing an orthonormal moving frame $\{\hat{\mathbf{E}}_A\}$ via the Gram--Schmidt process. First, normalize the first basis element,
%---------------------------------
\begin{equation}
\label{eq:i1}
	\hat{\mathbf{E}}_1
	\coloneq \frac{1}{\sqrt{\cM_{11}}}\,\frac{\partial}{\partial X^1}\,.
\end{equation}
%---------------------------------
Next, take the orthogonal projection of ${\partial}/{\partial X^2}$ relative to ${\partial}/{\partial X^1}$ which is given by
%---------------------------------
\begin{equation}
	\check{\mathbf{E}}_2
	\coloneq \frac{\partial}{\partial X^2} - \frac{\cM_{12}}{\cM_{11}}\,\frac{\partial}{\partial X^1}\,.
\end{equation}
%---------------------------------
Its squared norm is $\llangle \check{\mathbf{E}}_2,\check{\mathbf{E}}_2 \rrangle_{\M}=(\cM_{11}\cM_{22}-\cM_{12}^2)/\cM_{11}$, so that the normalized vector is
%---------------------------------
\begin{equation}
\label{eq:i2}
	\hat{\mathbf{E}}_2= \sqrt{ \frac{\cM_{11}}{\cM_{11}\cM_{22}-\cM_{12}^2}}~\check{\mathbf{E}}_2\,.
\end{equation}
%---------------------------------
Finally, let us take the orthogonal projection of ${\partial}/{\partial X^3}$ relative to the plane $(\partial/\partial X^1,\partial/\partial X^2)$
%---------------------------------
\begin{equation}
\begin{aligned}
	\check{\mathbf{E}}_3=	\frac{\partial}{\partial X^3}-	\frac{\cM_{13}}{\cM_{11}}\,\frac{\partial}{\partial X^1}
	-\frac{\cM_{11}\cM_{23}-\cM_{12}\cM_{13}}{\cM_{11}\cM_{22}-\cM_{12}^2}	
	\left(\frac{\partial}{\partial X^2}-\frac{\cM_{12}}{\cM_{11}}\,\frac{\partial}{\partial X^1}\right)\,.
\end{aligned}
\end{equation}
%---------------------------------
Its squared norm is $\llangle \check{\mathbf{E}}_3,\check{\mathbf{E}}_3 \rrangle_{\M} =\det\M/(\cM_{11}\cM_{22}-\cM_{12}^2)$, so that the normalized vector is
%---------------------------------
\begin{equation}
\label{eq:i3}
	\hat{\mathbf{E}}_3=\sqrt{\frac{\cM_{11}\cM_{22}-\cM_{12}^2}{\det\M}}~\check{\mathbf{E}}_3\,.
\end{equation}
%---------------------------------
The Gram--Schmidt process carried out above in~\eqref{eq:i1}--\eqref{eq:i3} may be compactly encoded in a linear transformation ${\Lt:T_X\mathcal{M}\to T_X\mathcal{M}}$ relating the construction of the orthonormal frame $\{\hat{\mathbf{E}}_A\}$ from the coordinate frame $\{\partial/\partial X^A\}$ as
%---------------------------------
\begin{equation}\label{eq:beta-def}
	\hat{\mathbf{E}}_A = \cL_A{}^B\,\frac{\partial}{\partial X^B}\,,
\end{equation}
%---------------------------------
with the orthonormalization transformation $\Lt$ explicitly given by the following representation
%---------------------------------
\begin{equation}
\label{eq:beta-explicit}
\renewcommand{\arraystretch}{1.3}
	\llbracket\cL_A{}^B\rrbracket
	=\begin{bmatrix}
	\dfrac{1}{\sqrt{\cM_{11}}} & 0 & 0 \\[8pt]
	\dfrac{-\cM_{12}}{\sqrt{\cM_{11}\,(\cM_{11}\cM_{22}-\cM_{12}^2)}}
	& \sqrt{\dfrac{\cM_{11}}{\cM_{11}\cM_{22}-\cM_{12}^2}} & 0 \\[12pt]
	\dfrac{\cM_{12}\cM_{23}-\cM_{13}\cM_{22}}{\sqrt{(\cM_{11}\cM_{22}-\cM_{12}^2)\det\M}}
	& \dfrac{\cM_{12}\cM_{13}-\cM_{11}\cM_{23}}{\sqrt{(\cM_{11}\cM_{22}-\cM_{12}^2)\det\M}}
	& \sqrt{\dfrac{\cM_{11}\cM_{22}-\cM_{12}^2}{\det\M}}
	\end{bmatrix}
\renewcommand{\arraystretch}{1}\,.
\end{equation}
%---------------------------------
The dual coframe $\{\hat{\boldsymbol\Omega}^A\}$ is defined by
$\langle \hat{\boldsymbol\Omega}^A,\hat{\mathbf{E}}_B \rangle=\delta^A_B$, and hence
%---------------------------------
\begin{equation} \label{eq:cobeta-def}
	\hat{\boldsymbol\Omega}^A = \big(\cL^{-1}\big)_B{}^A\,dX^B\,.
\end{equation}
%---------------------------------
As in the orthogonal case, the frame $\{\hat{\mathbf{E}}_A\}$ and coframe $\{\hat{\boldsymbol\Omega}^A\}$ are dimensionless. As a matter of fact, by construction $\llangle \hat{\mathbf{E}}_A,\hat{\mathbf{E}}_A \rrangle_{\M}=1$ (no summation on $A$), and hence $[\hat{\mathbf{E}}_A]^2=1$. The dual coframe inherits this property: $\langle \hat{\boldsymbol\Omega}^A,\hat{\mathbf{E}}_B \rangle=\delta^A_B$ gives $[\hat{\boldsymbol\Omega}^A]=1$.

Inverting \eqref{eq:beta-def} and its dual \eqref{eq:cobeta-def} yields
%---------------------------------
\begin{equation}
	\frac{\partial}{\partial X^B}=\big(\cL^{-1}\big)_B{}^A\,\hat{\mathbf{E}}_A\,,
	\quad
	dX^B=\cL_A{}^B\,\hat{\boldsymbol\Omega}^A\,.
\end{equation}
%---------------------------------
Substituting into the coordinate expressions for $\mathbf{U}=\mathrm{U}^A\,{\partial_A}$ and $\boldsymbol\vartheta=\vartheta_A\,dX^A$ gives
%---------------------------------
\begin{equation}
	\mathbf{U}
	=\mathrm{U}^B\,\frac{\partial}{\partial X^B}
	=\mathrm{U}^B\,\big(\cL^{-1}\big)_B{}^A\,\hat{\mathbf{E}}_A
	=\hat{\mathrm U}^A\,\hat{\mathbf{E}}_A\,,
	\quad
	\boldsymbol\vartheta
	=\vartheta_B\,dX^B
	=\vartheta_B\,\cL_A{}^B\,\hat{\boldsymbol\Omega}^A
	=\hat\vartheta_A\,\hat{\boldsymbol\Omega}^A\,,
\end{equation}
%---------------------------------
from which the \emph{physical components} of $\mathbf{U}$ and $\boldsymbol\vartheta$, i.e. their components relative to the (dimensionless) orthonormal frame $\{\hat{\mathbf{E}}_A\}$ and coframe $\{\hat{\boldsymbol\Omega}^A\}$, may be read off as
%---------------------------------
\begin{equation}\label{eq:Phys_Comp_gen}
	\hat{\mathrm U}^A = \big(\cL^{-1}\big)_B{}^A\,\mathrm{U}^B\,,\quad
	\hat\vartheta_A = \cL_A{}^B\,\vartheta_B\,.
\end{equation}
%---------------------------------
%---------------------------------
\begin{remark}[Transformation of physical components under a change of orthonormal frames]
\label{Remark:Frame-dep}
A metric fixes the family of orthonormal frames but not a specific member of it. In orthogonal coordinates (\S\ref{Subsection:orthogonal}) the normalized coordinate basis is a distinguished member---the unique frame tangent to the coordinate lines, whereas the Gram--Schmidt constructed frame above is merely one arbitrary choice, and any other orthonormal frame differs from it by an orthogonal transformation. Explicitly, any $\M$-orthonormal frame $\{\tilde{\mathbf E}_A\}$ is related to $\{\hat{\mathbf E}_A\}$ by a unique field of orthogonal transformations ${\mathbf T(X)\in O(T_X\mathcal{M},\M(X))}$ such that\footnote{$O(T_X\mathcal{M},\M(X))$ is the group of $\M$-orthogonal transformation over $T_X\mathcal{M}$, by which we mean a transformation ${\mathbf{T}(X):T_X\mathcal{M}\to T_X\mathcal{M}}$ that preserves the metric, i.e. $\mathbf{T} \M \mathbf{T}^\star = \M $; expressed in components with respect to the orthonormal frames $\{\tilde{\mathbf E}_A\}$ and $\{\hat{\mathbf E}_A\}$, this reads $\mathrm{T}_A{}^C \,\delta_{CD}\, \mathrm{T}_B{}^D = \delta_{AB}$.} 
%---------------------------------
\begin{equation}
\label{eq:Ortho_change}
	\tilde{\mathbf E}_A=\mathrm T_A{}^B\,\hat{\mathbf E}_B\,,
	\quad\textrm{or equivalently}\quad
	\hat{\mathbf E}_B=(\mathrm T^{-1})_B{}^A\,\tilde{\mathbf E}_A\,,
\end{equation}
%---------------------------------
and, conversely, every such $\mathbf{T}$ carries $\{\hat{\mathbf E}_A\}$ to an $\M$-orthonormal frame. The family of orthonormal frames adapted to $\M$ is thus a single equivalence class, any two members of which differing by an element of $O(T_X\mathcal{M},\M(X))$.\footnote{This gauge freedom is equivalently a freedom in the orthonormalization of the coordinate basis. From~\eqref{eq:beta-def} and~\eqref{eq:Ortho_change}, the transformed frame is generated by $\tilde{\Lt}=\mathbf{T} \Lt$, i.e. $\tilde{\mathbf E}_A=\tilde\cL_A{}^B\,\partial/\partial X^B$. Conversely, any orthonormalization map $\tilde{\Lt}$---e.g., Gram--Schmidt with a different ordering or a symmetric orthonormalization---satisfies the same orthonormality condition as $\Lt$, i.e. $\tilde{\Lt}\, \M\, \tilde{\Lt}^\star = \Lt \M \Lt^\star = \boldsymbol\delta$, which yields $\tilde{\Lt} \Lt^{-1} \M (\tilde{\Lt} \Lt^{-1})^\star = \boldsymbol\delta$, and hence defines an orthogonal transformation $\mathbf{T}\coloneq\tilde{\Lt} \Lt^{-1}$ relating $\tilde{\Lt}$ to $\Lt$. The choice of orthonormalization procedure is therefore indeed none other than the orthogonal gauge freedom~\eqref{eq:Ortho_change}.}
Fixing a new frame $\{\tilde{\mathbf E}_A\}$, it follows from \eqref{eq:Ortho_change} that
%---------------------------------
\begin{equation}
\mathbf{U}
	= \hat{\mathrm U}^A\,\hat{\mathbf E}_A
	= \hat{\mathrm U}^A (\mathrm T^{-1})_A{}^B\,\tilde{\mathbf E}_B
	= \tilde{\mathrm U}^B\,\tilde{\mathbf E}_B
	\,,
\end{equation}
%---------------------------------
where physical components of $\mathbf U$ with respect to the transformed orthonormal frame $\{\tilde{\mathbf E}_A\}$ are given by
%---------------------------------
\begin{equation}
\label{eq:vect_trans}
\tilde{\mathrm U}^B = (\mathrm T^{-1})_A{}^B\, \hat{\mathrm U}^A\,.
\end{equation}
%---------------------------------
Let $\{\tilde{\boldsymbol\Omega}^A\}$ denote the coframe dual to $\{\tilde{\mathbf E}_A\}$. Since $\tilde{\boldsymbol\Omega}^A(\tilde{\mathbf E}_B)=\delta^A_B$, one has
%---------------------------------
\begin{equation}
\label{eq:Co_Ortho_change}
	\tilde{\boldsymbol\Omega}^A=(\mathrm T^{-1})_B{}^A\,\hat{\boldsymbol\Omega}^B\,,
	\quad\textrm{or equivalently}\quad
	\hat{\boldsymbol\Omega}^B=\mathrm T_A{}^B\,\tilde{\boldsymbol\Omega}^A\,.
\end{equation}
%---------------------------------
Hence, one may write
%---------------------------------
\begin{equation}
\boldsymbol\vartheta
	= \hat\vartheta_A\, \hat{\boldsymbol\Omega}^A
	= \hat\vartheta_A\, \mathrm T_B{}^A\,\tilde{\boldsymbol\Omega}^B
	= \tilde\vartheta_B\, \tilde{\boldsymbol\Omega}^B
	\,,
\end{equation}
%---------------------------------
where physical components of $\boldsymbol\vartheta$ with respect to the transformed orthonormal coframe $\{\tilde{\boldsymbol\Omega}^A\}$ are given by
%---------------------------------
\begin{equation}
\label{eq:1-form_trans}
\tilde\vartheta_B = \mathrm T_B{}^A\, \hat\vartheta_A\,.
\end{equation}
%---------------------------------
Thus, given a choice of orthonormal frame physical components are uniquely determined, and they are otherwise defined up to the orthogonal transformations \eqref{eq:vect_trans} and \eqref{eq:1-form_trans}.
\end{remark}
%---------------------------------
%---------------------------------
\begin{remark}[Orthogonal coordinates]
\label{remark:Ortho_Coord_M}
In orthogonal coordinates the metric components $\cM_{AB}$ are diagonal, and the Gram--Schmidt construction $\Lt$ as given in \eqref{eq:beta-explicit}, and subsequently any other orthonormalization $\tilde{\Lt}=\mathbf{T} \Lt$, reduces to a diagonal transformation. In particular, one finds following \eqref{eq:beta-explicit} that
%---------------------------------
\begin{equation}
	\cL_A{}^B = \frac{1}{\sqrt{\cM_{AA}}}\,\delta_A^B\,,\quad
	\big(\cL^{-1}\big)_B{}^A = \sqrt{\cM_{AA}}\,\delta^A_B \quad (\text{no summation on $A$})\,.
	\label{eq:Ortho_Coord_M}
\end{equation}
%---------------------------------
Thus, the general formulae~\eqref{eq:Phys_Comp_gen} for physical components reduce to the orthogonal coordinate expressions~\eqref{eq:Phys_Comp}.
\end{remark}
%---------------------------------
%---------------------------------
\subsection{Physical components of two-point tensors}
\label{Subsection:two-point}
Let $\phi:\mathcal{M}\to\mathcal{N}$ be a smooth map between two Riemannian manifolds $(\mathcal{M},\M)$ and $(\mathcal{N},\m)$, and let $\mathbf{T}$ be a two-point tensor of type $\binom{r~r'}{s~s'}$ over $\phi$. Recall the coordinate representation~\eqref{Tensor-General-TwoPoint} in coordinate charts $\{X^A\}$ on $\mathcal{M}$ and $\{x^a\}$ on $\mathcal{N}$. Physical components now require orthonormalization on both legs---the domain tangent space $T_X\mathcal{M}$, equipped with $\M$, and the codomain tangent space $T_x\mathcal{N}$, equipped with $\m$. We work directly in general coordinates; the orthogonal-coordinate case follows by the diagonal reduction~\eqref{eq:Ortho_Coord_M} applied to each leg.

In general coordinates, the coordinate bases $\{\partial_A\}$ and $\{\partial_a\}$ are not necessarily orthogonal, and orthonormal frames may be constructed on each leg---e.g., by the Gram--Schmidt process of \S\ref{Subsection:general}---encoded in linear maps ${\Lt:T_X\mathcal{M}\to T_X\mathcal{M}}$ and ${\boldsymbol\ell:T_x\mathcal{N}\to T_x\mathcal{N}}$ such that
%---------------------------------
\begin{equation}
\label{eq:2pt-frames}
	\hat{\mathbf E}_A=\cL_A{}^B\,\frac{\partial}{\partial X^B}\,,
	\quad
	\hat{\mathbf e}_a=\ell_a{}^b\,\frac{\partial}{\partial x^b}\,,
\end{equation}
%---------------------------------
with dual coframes given by
%---------------------------------
\begin{equation}
\label{eq:2pt-coframes}
	\hat{\boldsymbol\Omega}^A=(\cL^{-1})_B{}^A\,dX^B\,,
	\quad
	\hat{\boldsymbol\omega}^a=(\ell^{-1})_b{}^a\,dx^b\,.
\end{equation}
%---------------------------------
As in \S\ref{Subsection:general}, the orthonormality conditions $\llangle\hat{\mathbf E}_A,\hat{\mathbf E}_B\rrangle_{\M}=\delta_{AB}$ and $\llangle\hat{\mathbf e}_a,\hat{\mathbf e}_b\rrangle_{\m}=\delta_{ab}$ render all four sets dimensionless. Inverting~\eqref{eq:2pt-frames} and its dual~\eqref{eq:2pt-coframes} yields
%---------------------------------
\begin{subequations}
\begin{alignat}{2}
	& \frac{\partial}{\partial X^B} =(\cL^{-1})_B{}^A\,\hat{\mathbf E}_A\,,
	\quad&
	& dX^B =\cL_A{}^B\,\hat{\boldsymbol\Omega}^A\,,
	\\
	& \frac{\partial}{\partial x^b} =(\ell^{-1})_b{}^a\,\hat{\mathbf e}_a\,,
	\quad&
	& dx^b =\ell_a{}^b\,\hat{\boldsymbol\omega}^a\,,
\end{alignat}
\end{subequations}
%---------------------------------
and substituting into~\eqref{Tensor-General-TwoPoint} gives physical components of $\mathbf{T}$
%---------------------------------
\begin{equation}\label{eq:Phys-Comp-2pt}
\begin{aligned}
	\hat T^{A_1\dots A_r\,a_1\dots a_{r'}}{}_{B_1\dots B_s\,b_1\dots b_{s'}}
	&= (\cL^{-1})_{C_1}{}^{A_1}\cdots(\cL^{-1})_{C_r}{}^{A_r}\,
	(\ell^{-1})_{c_1}{}^{a_1}\cdots(\ell^{-1})_{c_{r'}}{}^{a_{r'}}\\
	&\quad\times\,\cL_{B_1}{}^{D_1}\cdots\cL_{B_s}{}^{D_s}\,
	\ell_{b_1}{}^{d_1}\cdots\ell_{b_{s'}}{}^{d_{s'}}\,
	T^{C_1\dots C_r\,c_1\dots c_{r'}}{}_{D_1\dots D_s\,d_1\dots d_{s'}}\,.
\end{aligned}
\end{equation}
%---------------------------------
Each contravariant domain index is contracted with $\Lt^{-1}$ and each covariant domain index with $\Lt$, while each contravariant codomain index is contracted with $\boldsymbol\ell^{-1}$ and each covariant codomain index with $\boldsymbol\ell$. For a two-point tensor of type $\binom{0~1}{1~0}$---one contravariant codomain index and one covariant domain index, the type of the deformation gradient---this reduces to
%---------------------------------
\begin{equation}\label{eq:Phys-Comp-2pt-prototype}
	\hat T^a{}_A=(\ell^{-1})_b{}^a\,T^b{}_B\,\cL_A{}^B\,.
\end{equation}

%---------------------------------
%--------------------------------------------------------------
\paragraph{Change of orthonormal frames.}
The gauge freedom of Remark~\ref{Remark:Frame-dep} now acts on the two legs independently: each metric fixes the family of orthonormal frames on its tangent space but no specific member, so $T_X\mathcal{M}$ and $T_x\mathcal{N}$ each carry their own orthogonal gauge. For any pair of orthonormal frames $\{\tilde{\mathbf E}_A\}$ on $T_X\mathcal{M}$ and $\{\tilde{\mathbf e}_a\}$ on $T_x\mathcal{N}$, there exist $\mathbf Q(X)\in O(T_X\mathcal{M},\M)$ and $\mathbf q(x)\in O(T_x\mathcal{N},\m)$ such that
%---------------------------------
\begin{equation}
\label{eq:2pt-frame-change}
	\tilde{\mathbf E}_A=\mathrm Q_A{}^B\,\hat{\mathbf E}_B\,,
	\quad
	\tilde{\mathbf e}_a=\mathrm q_a{}^b\,\hat{\mathbf e}_b\,,
\end{equation}
%---------------------------------
with dual coframes
%---------------------------------
\begin{equation}
\label{eq:2pt-coframe-change}
	\tilde{\boldsymbol\Omega}^A = (\mathrm Q^{-1})_B{}^A\,\hat{\boldsymbol\Omega}^B\,,
	\quad
	\tilde{\boldsymbol\omega}^a = (\mathrm q^{-1})_b{}^a\,\hat{\boldsymbol\omega}^b\,,
\end{equation}
%---------------------------------
exactly as in~\eqref{eq:Ortho_change} and~\eqref{eq:Co_Ortho_change}. For the prototype~\eqref{eq:Phys-Comp-2pt-prototype} this gives us
%---------------------------------
\begin{equation}\label{eq:2pt-comp-change}
	\tilde T^a{}_A=(\mathrm q^{-1})_b{}^a\,\hat T^b{}_B\,\mathrm Q_A{}^B\,,
\end{equation}
%---------------------------------
the contravariant codomain index transforming as a vector~\eqref{eq:vect_trans} and the covariant domain index as a $1$-form~\eqref{eq:1-form_trans}; the general case~\eqref{eq:Phys-Comp-2pt} transforms one such factor per index. Thus, for a fixed pair of metrics $(\M,\m)$, physical components of a two-point tensor are uniquely defined up to independent orthogonal transformations of the two frames, transforming under $O(T_X\mathcal{M},\M)\times O(T_x\mathcal{N},\m)$. As in Remark~\ref{Remark:Frame-dep}, choosing the two frames is equivalent to choosing an orthonormalization procedure on each leg.
%-----------------------------
\begin{remark}[Physical components as frame-dependent representations]
\label{Remark:Frame-dependence}
A tensor field is an intrinsic geometric object, defined independently of any choice of coordinates, basis, or frame. Once metrics are chosen on $\mathcal{M}$ and $\mathcal{N}$, one may represent a two-point tensor $\mathbf{T}$ with respect to orthonormal frames on its two legs and obtain what are commonly called its physical components. These components are not intrinsic. The reason is that a metric determines the family of orthonormal frames but not a unique member of it, so that physical components are fixed only up to the gauge freedom~\eqref{eq:2pt-comp-change}. For orthogonal coordinate systems there is a distinguished orthonormal frame, obtained by normalizing the coordinate basis vectors; this provides a natural choice and explains the widespread use of physical components in applications. The choice is tied to the coordinate system, however, and not to the metric, so it is not intrinsic either. In general coordinates there is no preferred frame, and hence no canonical notion of physical components.
This situation is completely analogous to the representation of a vector in Euclidean space. A vector is an intrinsic geometric object, whereas its Cartesian components depend on the chosen orthonormal basis. The metric identifies which bases are orthonormal, but it does not single out a preferred one. Likewise, physical components are simply the components of a tensor relative to a chosen orthonormal frame.
The tensor itself is thus the fundamental geometric object, its physical components being frame-dependent representations of it---representations one can hardly do without, being referred throughout to unit directions and carrying, unlike the coordinate components, a single physical dimension. 
The point is not to deny their usefulness, but to emphasize their non-uniqueness.
\end{remark}
%-----------------------------
%---------------------------------
\subsection{Normalized and physical components: A general framework}
\label{Subsection:generalization}
The orthonormal-frame constructions of the preceding subsections are special cases of a broader family of schemes that define tensor components with consistent physical dimension. In this subsection, we characterize this family, whose members we call normalization transformations and whose outputs we call normalized components. We identify the additional requirements under which normalized components represent physically significant quantities, and hence deserve the designation \emph{physical}, and determine what these requirements entail in general coordinates. We then use this framework to assess alternative constructions proposed in the literature, notably that of \citet{Aris1962}, against the same requirements.
%---------------------------------
%---------------------------------
\subsubsection{Dimensional consistency}
\label{Subsubsection:Dimensional-Consistency}
Let $\{X^I\}$ be an arbitrary local coordinate chart on $(\mathcal M,\M)$, with $n=\dim\mathcal M$. Let $\boldsymbol{\mathsf{A}}$ and $\boldsymbol{\mathsf{B}}$ be invertible, position-dependent transformations acting on the coordinate frame and coframe, respectively, and generating the frame and coframe fields
%---------------------------------
\begin{equation}\label{eq:AB-frames}
	\hat{\mathbf E}_A
	=\mathsf{A}_A{}^I\,\partial_I\,,
	\quad
	\hat{\boldsymbol\Omega}^A
	=\mathsf{B}_I{}^A\,dX^I\,,
\end{equation}
%---------------------------------
with no orthonormality, or any relation between the two assumed at this stage. An arbitrary vector $\mathbf U=\mathrm{U}^I\,\partial_I$ and $1$-form $\boldsymbol\vartheta=\vartheta_I\,dX^I$ have the representations $\mathbf U=\hat{\mathrm U}^A\,\hat{\mathbf E}_A$ and $\boldsymbol\vartheta=\hat\vartheta_A\,\hat{\boldsymbol\Omega}^A$, with components
%---------------------------------
\begin{equation}\label{eq:AB-components}
	\hat{\mathrm U}^A
	=(\mathsf{A}^{-1})_I{}^A\,\mathrm{U}^I\,,
	\quad
	\hat\vartheta_A
	=(\mathsf{B}^{-1})_A{}^I\,\vartheta_I\,.
\end{equation}
%---------------------------------
We call $\boldsymbol{\mathsf{A}}$ (and $\boldsymbol{\mathsf{B}}$, respectively) a \emph{normalization transformation} if $[\hat{\mathrm U}^A]=[\mathbf U]$ for every vector $\mathbf U$ (if $[\hat\vartheta_A]=[\boldsymbol\vartheta]$ for every $1$-form $\boldsymbol\vartheta$, respectively), i.e. if all the components it defines carry the physical dimension of the underlying tensor field.
%-----------------------------
\begin{lem}[Dimensional characterization]\label{Lemma:Dimensional}
$\boldsymbol{\mathsf{A}}$ is a normalization transformation if and only if the frame it generates is dimensionless, i.e. $[\hat{\mathbf E}_A]=1$. Similarly, $\boldsymbol{\mathsf{B}}$ is a normalization transformation if and only if the coframe it generates is dimensionless, i.e. $[\hat{\boldsymbol\Omega}^A]=1$.
\end{lem}
%-----------------------------
\begin{proof}
Suppose $\boldsymbol{\mathsf{A}}$ is a normalization transformation and take $\mathbf U=\hat{\mathbf E}_B$, whose components are $\hat{\mathrm U}^A=\delta^A_B$ and hence dimensionless; then $[\hat{\mathbf E}_B]=[\mathbf U]=[\hat{\mathrm U}^A]=1$. Conversely, $\mathbf U=\hat{\mathrm U}^A\,\hat{\mathbf E}_A$ gives $[\mathbf U]=[\hat{\mathrm U}^A]\,[\hat{\mathbf E}_A]$ (no summation on $A$), so that $[\hat{\mathbf E}_A]=1$ yields $[\hat{\mathrm U}^A]=[\mathbf U]$ for every $\mathbf U$. The $1$-form proof follows similarly.
\end{proof}
%-----------------------------

Dimensional consistency, however, is not a requirement imposed upon normalization transformations but the attribute that defines them, and it does not by itself suffice. 
Two additional requirements are needed for the normalized components to be physically significant: first, normalization must preserve the component expression of the natural pairing, so that $\langle\boldsymbol{\vartheta},\mathbf{U}\rangle=\vartheta_I \mathrm{U}^I=\hat{\vartheta}_A\,\hat{\mathrm  U}^A$ for every vector $\mathbf{U}$ and $1$-form $\boldsymbol{\vartheta}$. Second, the normalized frame and coframe must be of unit length. Without this requirement, their arbitrary rescaling produces corresponding inverse rescalings of the normalized components, so the numerical value of a component depends on the scale assigned to the frame or coframe element rather than only on the tensor and the associated direction. Unit normalization removes this scaling ambiguity, and one requires that ${\big\lVert\hat{\mathbf E}_A\big\rVert_{\M}=1}$ and ${\big\lVert\hat{\boldsymbol\Omega}^A\big\rVert_{\M^\sharp}=1}$.
Neither requirement follows from dimensional consistency alone.
The first may fail because the maps $\boldsymbol{\mathsf{A}}$ and $\boldsymbol{\mathsf{B}}$ may be chosen independently, so that the normalizations of contravariant and covariant components need not be mutually consistent. 
The second may fail because the class of normalization transformations is closed under $\boldsymbol{\mathsf{A}}\mapsto\boldsymbol{\mathsf{S}}\,\boldsymbol{\mathsf{A}}$ and $\boldsymbol{\mathsf{B}}\mapsto\boldsymbol{\mathsf{B}}\,\boldsymbol{\mathsf{T}}$ for arbitrary dimensionless invertible transformations $\boldsymbol{\mathsf{S}}(X),\boldsymbol{\mathsf{T}}(X)\in GL(T_X\mathcal{M})$,\footnote{$GL(T_X\mathcal{M})$ is the group of invertible linear transformations of $T_X\mathcal{M}$; expressed in components with respect to a pair of frames, it is the group of invertible matrices $\llbracket\mathsf{S}_A{}^B\rrbracket$. Note that the orthogonal group discussed earlier is the subgroup $O(T_X\mathcal{M},\M)\subset GL(T_X\mathcal{M})$ preserving the metric $\M$.} and any two sets $(\boldsymbol{\mathsf{A}},\boldsymbol{\mathsf{B}})$ and $(\boldsymbol{\mathsf{A}}',\boldsymbol{\mathsf{B}}')$ of normalization transformations differ by exactly one such dimensionless pair, $\boldsymbol{\mathsf{S}}=\boldsymbol{\mathsf{A}}'\boldsymbol{\mathsf{A}}^{-1}$ and $\boldsymbol{\mathsf{T}}=\boldsymbol{\mathsf{B}}^{-1}\boldsymbol{\mathsf{B}}'$.\footnote{$\boldsymbol{\mathsf{S}}=\boldsymbol{\mathsf{A}}'\boldsymbol{\mathsf{A}}^{-1}$ and $\boldsymbol{\mathsf{T}}=\boldsymbol{\mathsf{B}}^{-1}\boldsymbol{\mathsf{B}}'$ are dimensionless as they effectively relate dimensionless frames and coframes.} Each map is thus determined only up to a dimensionless $GL(T_X\mathcal{M})$ gauge freedom, the frame and the coframe being rescalable at will.\footnote{The independent gauge freedoms associated with $\boldsymbol{\mathsf{S}}$ and $\boldsymbol{\mathsf{T}}$ further emphasize that dimensional consistency alone does not ensure compatibility between the normalized contravariant and covariant components.}
In what follows, we impose the two requirements in turn, and record at each step the reduction of this gauge freedom.
%---------------------------------
%---------------------------------
\subsubsection{Dual compatibility}
The first requirement noted above---that of natural pairing consistency---may be achieved by requiring the normalization to preserve the natural pairing between vectors and $1$-forms, i.e. 
%---------------------------------
\begin{equation}\label{Normalization-Compatibility}
	(\text{D}):\ \hat\vartheta_A\,\hat{\mathrm U}^A=\vartheta_I\,\mathrm{U}^I \,,
\end{equation}
%---------------------------------
for every vector $\mathbf U$ and covector $\boldsymbol\vartheta$. From \eqref{eq:AB-components}, and since $\langle\hat{\boldsymbol\Omega}^B,\hat{\mathbf E}_A\rangle=\mathsf{A}_A{}^I\,\mathsf{B}_I{}^B$, this requirement admits the equivalent characterizations
%---------------------------------
\begin{equation}\label{Equi-Normalization-Compatibility}
	\hat\vartheta_A\,\hat{\mathrm U}^A=\vartheta_I\,\mathrm{U}^I
	\quad\Longleftrightarrow\quad
	\boldsymbol{\mathsf{B}}=\boldsymbol{\mathsf{A}}^{-1}
	\quad\Longleftrightarrow\quad
	\langle\hat{\boldsymbol\Omega}^B,\hat{\mathbf E}_A\rangle=\delta_A^B\,.
\end{equation}
%---------------------------------
A pair $(\boldsymbol{\mathsf{A}},\boldsymbol{\mathsf{B}})$ satisfying~\eqref{Normalization-Compatibility}---or any of the equivalent conditions of \eqref{Equi-Normalization-Compatibility}---is said to be \emph{dual-compatible} and may hence be designated by a single transformation. Dual compatibility is a characterization rather than a mere constraint: the last of \eqref{Equi-Normalization-Compatibility} states that $\{\hat{\boldsymbol\Omega}^A\}$ is the coframe dual to $\{\hat{\mathbf E}_A\}$, i.e. that the two normalization rules descend from a single frame field rather than from two unrelated ones. The component formulae \eqref{eq:AB-components} are then governed by the single map $\boldsymbol{\mathsf{A}}$,
%---------------------------------
\begin{equation}\label{eq:DC-components}
	\hat{\mathrm U}^A
	=(\mathsf{A}^{-1})_I{}^A\,\mathrm{U}^I\,,
	\quad
	\hat\vartheta_A
	=\mathsf{A}_A{}^I\,\vartheta_I\,,
\end{equation}
%---------------------------------
and, more generally, for a tensor $\mathbf T=T^{I_1\ldots I_r}{}_{J_1\ldots J_s}\,\partial_{I_1}\otimes\cdots\otimes\partial_{I_r}\otimes dX^{J_1}\otimes\cdots\otimes dX^{J_s}$, the normalized components
%---------------------------------
\begin{equation}\label{eq:general-normalized-tensor}
	\hat{T}^{A_1\ldots A_r}{}_{B_1\ldots B_s}
	=(\mathsf{A}^{-1})_{I_1}{}^{A_1}\cdots(\mathsf{A}^{-1})_{I_r}{}^{A_r}\,
	\mathsf{A}_{B_1}{}^{J_1}\cdots\mathsf{A}_{B_s}{}^{J_s}\,
	T^{I_1\ldots I_r}{}_{J_1\ldots J_s}\,,
\end{equation}
%---------------------------------
are obtained by contracting each contravariant index with $\boldsymbol{\mathsf{A}}^{-1}$ and each covariant index with $\boldsymbol{\mathsf{A}}$. The two independent gauge freedoms collapse into one: $\boldsymbol{\mathsf{B}}$ is determined by $\boldsymbol{\mathsf{A}}$, the two transformations being locked by $\boldsymbol{\mathsf{T}}=\boldsymbol{\mathsf{S}}^{-1}$, and a single dimensionless $GL(T_X\mathcal{M})$ gauge freedom remains---precisely the condition under which the pairing is gauge-invariant.
%---------------------------------
%---------------------------------
\subsubsection{Unit length and orthonormality}
The second requirement noted above, which consists of removing the freedom to rescale the frame and coframe, and hence the resulting ambiguity in the magnitudes of the normalized components, may be achieved by setting the frame and coframe to be of unit length:
%---------------------------------
\begin{subequations}\label{eq:unit-length}
\begin{alignat}{3}
	&\text{(N)}\colon
	&\quad
	\big\lVert\hat{\mathbf E}_A\big\rVert_{\M}^{2}
	&=\mathsf{A}_A{}^I\,\mathsf{A}_A{}^J\,\cM_{IJ}=1\,,
	&\quad
	&\text{(no summation on $A$)}\,,
	\label{eq:unit-length-N}
	\\
	&\text{(N}^*\text{)}\colon
	&\quad
	\big\lVert\hat{\boldsymbol\Omega}^A\big\rVert_{\M^\sharp}^{2}
	&=\mathsf{B}_I{}^A\,\mathsf{B}_J{}^A\,\cM^{IJ}=1\,,
	&\quad
	&\text{(no summation on $A$)}\,.
	\label{eq:unit-length-Nstar}
\end{alignat}
\end{subequations}
%---------------------------------
Under dual compatibility, dimensional consistency is equivalently set either by the frame or its dual. Unit length, however, is not similarly inherited, since duality is a statement on the natural pairing and not the metric. As a matter of fact, the coframe dual to a frame of unit length need not, in general, be of unit length. 
Conditions (N) and (N$^*$) must therefore be imposed separately. The following proposition characterizes when they can hold simultaneously under dual compatibility.
%-----------------------------
\begin{prop}\label{Prop:Rigidity}
For a dual-compatible normalization transformation $\boldsymbol{\mathsf{A}}$, the following are equivalent:
	(i) \emph{(N)} \& \emph{(N$^*$)} hold simultaneously;
	(ii) $\{\hat{\mathbf E}_A\}$ is $\M$-orthonormal;
	(iii) $\{\hat{\boldsymbol\Omega}^A\}$ is $\M^\sharp$-orthonormal.
\end{prop}
%-----------------------------
\begin{proof}
Let $\mathsf{H}_{AB}\coloneq\llangle\hat{\mathbf E}_A,\hat{\mathbf E}_B\rrangle_{\M}=\mathsf{A}_A{}^I\,\mathsf{A}_B{}^J\,\cM_{IJ}$ denote the Gram matrix of the frame, a symmetric positive-definite matrix. By dual compatibility ($\boldsymbol{\mathsf{B}}=\boldsymbol{\mathsf{A}}^{-1}$), and the Gram matrix of the coframe reads
%---------------------------------
\begin{equation}\label{eq:coframe-gram}
	\llangle\hat{\boldsymbol\Omega}^A,
	\hat{\boldsymbol\Omega}^B\rrangle_{\M^\sharp}
	=(\mathsf{A}^{-1})_I{}^A\,\mathsf{A}^{-1})_J{}^B\,\cM^{IJ}
	=(\mathsf{H}^{-1})^{AB}\,.
\end{equation}
%---------------------------------
The equivalence of (ii) and (iii) is then immediate, $\llbracket\mathsf{H}_{AB}\rrbracket=\llbracket\delta_{AB}\rrbracket$ if and only if $\llbracket(\mathsf{H}^{-1})^{AB}\rrbracket=\llbracket\delta^{AB}\rrbracket$, and either implies (i). Conversely, let $\lambda_1,\ldots,\lambda_n>0$ denote the eigenvalues of $\llbracket\mathsf{H}_{AB}\rrbracket$. Conditions (N) and (N$^*$) fix the diagonal entries of $\llbracket\mathsf{H}_{AB}\rrbracket$ and $\llbracket(\mathsf{H}^{-1})^{AB}\rrbracket$ to unity by \eqref{eq:coframe-gram}, and hence both traces to $n$, so that
%---------------------------------
\begin{equation}
	\label{eq:sum_eigen}
	\sum_{k=1}^{n}\Big(\lambda_k+\frac{1}{\lambda_k}\Big)
	=\operatorname{tr}\llbracket\mathsf{H}_{AB}\rrbracket
	+\operatorname{tr}\llbracket(\mathsf{H}^{-1})^{AB}\rrbracket
	=2n\,.
\end{equation}
%---------------------------------
Since $\lambda+\lambda^{-1}\geqslant2$ for every $\lambda>0$, with equality if and only if $\lambda=1$, it follows from~\eqref{eq:sum_eigen} that every eigenvalue equals unity. Real and symmetric, $\llbracket\mathsf{H}_{AB}\rrbracket$ is orthogonally diagonalizable, $\llbracket\mathsf{H}_{AB}\rrbracket=\llbracket\mathsf{Q}\rrbracket\,\llbracket\Lambda\rrbracket\,\llbracket\mathsf{Q}\rrbracket^{\top}$ with $\llbracket\mathsf{Q}\rrbracket$ orthogonal and $\llbracket\Lambda\rrbracket=\llbracket\delta_{AB}\rrbracket$ the diagonal matrix of its eigenvalues. Hence, $\llbracket\mathsf{H}_{AB}\rrbracket=\llbracket\mathsf{Q}\rrbracket\,\llbracket\mathsf{Q}\rrbracket^{\top}=\llbracket\delta_{AB}\rrbracket$, and the frame is $\M$-orthonormal.
\end{proof}
%-----------------------------
%---------------------------------
\subsubsection{From normalized components to physical components}
Together with the inherent dimensional consistency of normalization transformations, the two additional requirements justify the designation \emph{physical components}: dimensional consistency gives the components the same physical dimensions as the tensor field, dual compatibility ensures consistency between the normalized contravariant and covariant components, and unit normalization of the frame and coframe removes the arbitrary scaling of the components associated with rescaling their frame and coframe elements.
By Proposition~\ref{Prop:Rigidity}, the physically admissible transformations are precisely those dual-compatible that orthonormalize the coordinate frame (${\mathsf{A}_A{}^I\,\mathsf{A}_B{}^J\,\cM_{IJ}=\delta_{AB}}$), and the residual $GL(T_X\mathcal{M})$ ambiguity is thereby reduced to the orthogonal subgroup $O(T_X\mathcal{M},\M(X))$: every admissible transformation is of the form $\mathsf{A}_A{}^I=\mathsf{T}_A{}^B\,\cL_B{}^I$ for some $\mathsf{T}(X)\in O(T_X\mathcal{M},\M(X))$, the Gram--Schmidt map $\Lt$ of~\eqref{eq:beta-explicit} being one representative, singled out by the construction rather than by the conditions and no more distinguished than any other. The orthogonal gauge freedom of Remark~\ref{Remark:Frame-dep} is thus not an artifact of that construction: it is what remains of the ambiguity of the general framework once all three conditions are imposed.
%-----------------------------
\begin{remark}[The construction of Aris and its dual]
\label{Remark:Aris}
\citet{Aris1962} extended physical components to non-orthogonal coordinates by normalizing each coordinate basis vector,
%---------------------------------
\begin{equation}
\label{eq:Aris-frame}
	\hat{\mathbf E}_{(A)}
	=\frac{1}{\sqrt{\cM_{AA}}}\,\partial_A
	\quad\text{(no summation on $A$)}\,,
\end{equation}
%---------------------------------
so that for a vector $\mathbf U=\mathrm{U}^I\,\partial_I$, physical components are defined by
%---------------------------------
\begin{equation}\label{eq:Aris_contr}
	\mathrm U_{(A)}=\sqrt{\cM_{AA}}\,\mathrm{U}^A
	\quad\text{(no summation on $A$)}\,.
\end{equation}
%---------------------------------
For a covector ($1$-form) $\boldsymbol\vartheta=\vartheta_I\,dX^I$, one first raises the index and then applies the same rule,
%---------------------------------
\begin{equation}\label{eq:Aris_cov}
	\vartheta_{(A)}=\sqrt{\cM_{AA}}\,\cM^{AI}\vartheta_I
	\quad\text{(no summation on $A$)}\,.
\end{equation}
%---------------------------------
In the present notation, \eqref{eq:Aris_contr} and \eqref{eq:Aris_cov} correspond to the pair
%---------------------------------
\begin{equation}
\label{eq:Aris-AB}
	\mathsf{A}_A{}^I=\frac{1}{\sqrt{\cM_{AA}}}\,\delta_A^I\,,
	\quad
	\mathsf{B}_I{}^A=\frac{1}{\sqrt{\cM_{AA}}}\,\cM_{IA}
	\quad\text{(no summation on $A$)}\,,
\end{equation}
%---------------------------------
whose generated coframe is
%---------------------------------
\begin{equation}\label{eq:Aris-coframe}
	\hat{\boldsymbol\Omega}^A
	=\frac{1}{\sqrt{\cM_{AA}}}\,\cM_{IA}\,dX^I
	=\big(\hat{\mathbf E}_{(A)}\big)^{\flat}
	\quad\text{(no summation on $A$)}\,;
\end{equation}
%---------------------------------
Aris's covector rule is his vector rule applied to $\boldsymbol\vartheta^\sharp$, and the coframe it generates is the metric-dual of his frame rather than its dual coframe. The pair \eqref{eq:Aris-AB} consists of normalization transformations and satisfies (N) and (N$^*$)---the musical isomorphisms being isometries---but it is not dual-compatible:
%---------------------------------
\begin{equation}
\label{eq:Aris-pairing}
	\langle\hat{\boldsymbol\Omega}^B,\hat{\mathbf E}_{(A)}\rangle
	=\frac{\cM_{AB}}{\sqrt{\cM_{AA}\,\cM_{BB}}} 
	\quad\text{(no summation on $A$ and $B$)}\,.
\end{equation}
%---------------------------------
Note that, by Proposition~\ref{Prop:Rigidity}, a pair satisfying (D), (N), and (N$^*$) would render the coordinate directions orthogonal. The underlying reason is that the flat map preserves length while duality preserves the pairing, the two coinciding precisely on an orthonormal frame; a construction referred to normalized non-orthogonal coordinate directions must therefore choose between them. Retaining Aris's frame $\{\hat{\mathbf E}_{(A)}\}$ from~\eqref{eq:Aris-frame} but pairing it with its genuine dual coframe restores dual compatibility (D) and vector normalization (N), while sacrificing covector normalization (N$^*$)---the dual coframe elements failing, in general, to be of unit length by \eqref{eq:coframe-gram} whenever the chart is non-orthogonal.

The mirror construction, which starts instead from the normalized coordinate coframe,
%---------------------------------
\begin{equation}
\label{eq:mirror-coframe}
	\hat{\boldsymbol\Omega}^A
	=\frac{1}{\sqrt{\cM^{AA}}}\,dX^A\,,
	\quad
	\hat{\mathbf E}_{(A)}=\big(\hat{\boldsymbol\Omega}^A\big)^{\sharp}
	\quad\text{(no summation on $A$)}\,,
\end{equation}
%---------------------------------
i.e. $\hat\vartheta_A=\sqrt{\cM^{AA}}\,\vartheta_A$ and $\hat{\mathrm U}^A=\sqrt{\cM^{AA}}\,\cM_{AI}\,\mathrm{U}^I$ (no summation on $A$), satisfies (N) and (N$^*$) but fails (D) for the same reason as Aris's original construction. Note that it does not coincide with it, so that relinquishing dual compatibility does not single out an alternative. Retaining the coframe $\{\hat{\boldsymbol\Omega}^A\}$ from~\eqref{eq:mirror-coframe} but pairing it with its genuine dual frame instead of its metric-dual restores (D) and (N$^*$) while sacrificing (N)---the dual frame elements failing, in general, to be of unit length, by the same argument as in \eqref{eq:coframe-gram}, whenever the chart is non-orthogonal. Each of the four constructions satisfies two of the three conditions; none satisfies all three (D), (N), and (N$^*$).
\end{remark}
%-----------------------------

In summary, if physical components are understood merely as normalized components---that is, components carrying the physical dimensions of the underlying tensor field and nothing further---they remain ambiguous under independent dimensionless $GL(T_X\mathcal{M})$ gauge transformations acting on the contravariant and covariant indices. Requiring the natural pairing to retain its component form relates these two gauge freedoms, leaving a single $GL(T_X\mathcal{M})$ gauge freedom. Requiring, in addition, both the normalized frame and its dual coframe to be of unit length reduces the remaining freedom to the orthogonal subgroup $O(T_X\mathcal{M},\M(X))$. Together, these requirements reduce the general normalization framework to the orthonormal-frame construction of \S\ref{Subsection:general}. Its residual orthogonal gauge freedom is therefore not an artifact of the Gram--Schmidt procedure, but precisely the ambiguity left undetermined by dimensional consistency, dual compatibility, and unit normalization. Constructions retaining greater freedom necessarily relinquish at least one of these requirements, as illustrated by the alternatives in Remark~\ref{Remark:Aris}: they sacrifice either compatibility between the normalized contravariant and covariant components or unit normalization of the frame or coframe. Thus, up to a local orthogonal transformation, the construction adopted in this paper is the only one whose components have the same physical dimensions as the tensor field, preserve the component expression of the natural pairing, and are defined with respect to a unit-normalized frame and coframe. The remaining orthogonal ambiguity cannot be eliminated without introducing an additional choice of frame not determined by these requirements. Orthogonal coordinates provide such a preferred frame through their coordinate directions, whereas general coordinates do not.
%-----------------------------
%-----------------------------
\section{Physical components in elasticity and inelasticity}
\label{Section:Anelasticity}
In this section, we first briefly discuss the geometric framework of nonlinear elasticity and anelasticity and then proceed to examine physical components, as introduced in \S\ref{Section:Physical-Components}, within the setting of this framework. Inelastic processes beyond anelasticity are treated in \S\ref{Sec:Beyond-Anelasticity}, where viscoelasticity and visco-anelasticity are worked out.
%--------------------------------------------------------------
\subsection{Overview of geometric elasticity and anelasticity}
\label{S:overview}
In what follows, we give a terse account on nonlinear elasticity and anelasticity---e.g., \citep{MarsdenHughes1983, YavariSozio2023, SaYa2025GenColemanNoll}.
%--------------------------------------------------------------
\paragraph{Body, ambient space, and motion.}
We model the elastic body as a smooth three-dimensional manifold $\mathcal{B}$ embedded in the Euclidean ambient space $\mathcal{S} = \mathbb{R}^3$. Its motion is described by a time-dependent deformation map $\varphi_t : \mathcal{B} \to \mathcal{S}$, which carries each material point $X \in \mathcal{B}$ to a spatial location $x = \varphi_t(X) \in \mathcal{S}$. With the motion separating the referential and spatial pictures, we label quantities in the two accordingly: referential objects are denoted by uppercase characters and spatial ones with lowercase characters. We work in local charts $\{X^A\}$ on $\mathcal{B}$ and $\{x^a\}$ on $\mathcal{S}$, with corresponding coordinate frames $\{\partial/\partial X^A\}$ and $\{\partial/\partial x^a\}$ and corresponding reciprocal coframes $\{dX^A\}$ and $\{dx^a\}$. Repeated indices are summed in the sense of Einstein notation, i.e. $\mathrm{u}^i \mathrm v_i \coloneq \sum_i \mathrm{u}^i \mathrm v_i$.
The ambient space is equipped with the flat (Euclidean) metric $\g = \cg_{ab}\, dx^a \otimes dx^b$, whose components reduce to the Kronecker delta, $\cg_{ab} = \delta_{ab}$, in Cartesian coordinates. For two spatial vectors $\mathbf{u}, \mathbf{w} \in T_x\mathcal{S}$, the metric inner product is $\llangle \mathbf{u}, \mathbf{w} \rrangle_{\g} =\mathrm{u}^a \mathrm{w}^b \cg_{ab}$, while the duality pairing of a $1$-form $\boldsymbol{\alpha} \in T^*_x\mathcal{S}$ with a vector reads $\langle \boldsymbol{\alpha}, \mathbf{u} \rangle = \alpha_a \mathrm{u}^a$. The associated Riemannian volume form is $dv = \sqrt{\det \g}\, dx^1 \wedge dx^2 \wedge dx^3$, and we write $\bar\nabla$ for the Levi-Civita connection of $(\mathcal{S}, \g)$, with Christoffel symbols ${\gamma^a}_{bc}$.
%--------------------------------------------------------------
\paragraph{Euclidean reference configuration.}
Pulling the ambient metric back to the body through the inclusion map $\iota: \mathcal{B} \hookrightarrow \mathcal{S}$ endows $\mathcal{B}$ with a Euclidean reference metric $\Go = \iota^* \g$; equivalently, $\Go$ is the restriction of $\g$ to $\mathcal{B}$. Componentwise, $\Go = \cGo_{AB}\, dX^A \otimes dX^B$.
For instance, $\Go = dR \otimes dR + R^2\, d\Theta \otimes d\Theta + dZ \otimes dZ$ in cylindrical coordinates.
It should be stressed that $\Go$ records only the geometry that $\mathcal{B}$ inherits from its embedding and need not coincide with the body's intrinsic, stress-free geometry, which anelastic eigenstrains may render different. The $\Go$-inner product and its duality pairing on $T_X\mathcal{B}$ are denoted similarly to their spatial analogues, the volume form is $d\mathring{V} = \sqrt{\det \Go}\, dX^1 \wedge dX^2 \wedge dX^3$, and the Levi-Civita connection of $(\mathcal{B}, \Go)$ is denoted $\mathring{\nabla}$, with Christoffel symbols ${\mathring{\Gamma}^A}_{BC}$.
%--------------------------------------------------------------
\paragraph{Deformation gradient.}
The local deformation may be quantified by the deformation gradient, defined as the tangent (derivative) map of the motion $\varphi$; it reads $\F(X,t) = T\varphi_t(X)$ over the domains $T_X\mathcal{B} \to T_{\varphi_t(X)}\varphi_t(\mathcal{B})$. Being the differential of a map between two manifolds, $\F$ is a two-point tensor: it sends material tangent vectors to spatial ones. Its components are given by the partial derivatives of the motion yielding the following coordinate representation ${\F(X) = \cF^a{}_A(X)\, \partial_a \otimes dX^A}$.

%--------------------------------------------------------------
\paragraph{Material configuration in finite anelasticity.}
%-----------------------------
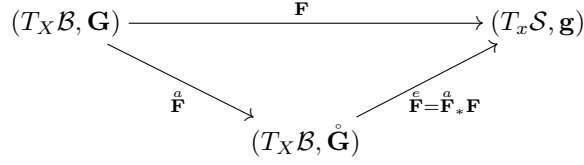
\begin{figure}
\centerline{%
\xymatrix@C=1.5cm@R=1.0cm{
(T_X\mathcal{B},\mathbf{G})
\ar[rr]^{\mathbf{F}}
\ar[dr]_{\Fa}
&&
(T_x\mathcal{S},\mathbf{g})
\\
&
(T_X\mathcal{B},\mathring{\mathbf{G}})
\ar[ur]_{\Fe=\Fa_*\mathbf{F}}
&
}}
\caption[Elastic and anelastic distortions.]{Elastic and anelastic distortions and their associated metrics. The anelastic distortion $\Fa$ is an isometry from the material configuration $(\mathcal{B},\G)$ onto the Euclidean configuration $(\mathcal{B},\Go)$. Physical components of $\F$ are referred to $\G$, those of $\Fe$ to $\Go$, and the spatial leg of both is referred to $\g$.}
\label{Fig:distortions}
\end{figure}
%-----------------------------
In the presence of eigenstrains\textemdash intrinsic distortions arising from plasticity, growth, swelling, or thermal expansion\textemdash the natural configuration of the body generally admits no isometric embedding in the Euclidean ambient space, and relaxing the body in physical space leaves it residually stressed. Residual stress is thus an expression of the non-Euclidean character of the natural configuration, which must accordingly be endowed with a material metric $\G$ distinct from $\Go$. Anelastic effects are built into the kinematics through the Bilby--Kr\"oner--Lee (BKL) multiplicative split of the deformation gradient \citep{bilby1955, kroner1959, leeliu1967, Sadik2017} $\F = \Fe \Fa$,
where the anelastic factor $\Fa: T_X\mathcal{B} \to T_X\mathcal{B}$ encodes the local, stress-relieving distortion produced by the eigenstrains, mapping the Euclidean reference configuration $(\mathcal{B}, \Go)$ to the stress-free material configuration $(\mathcal{B}, \G)$, whereas the elastic factor $\Fe: T_X\mathcal{B} \to T_x\mathcal{S}$ describes the recoverable elastic deformation from the material configuration to the current one---see Fig.~\ref{Fig:distortions}. While the total deformation gradient is compatible, $\F = T\varphi$, the individual distortions $\Fe$ and $\Fa$ are in general incompatible\textemdash they are not derivatives of smooth mappings\textemdash and this incompatibility is precisely the failure of the locally relaxed state to be assembled into a global embedding in physical space. The natural, stress-free lengths reside on the resulting abstract \emph{material manifold}, whose geometry is carried by the material metric obtained by pulling $\Go$ back along $\Fa$,\footnote{To see this, consider a curve $\gamma: \mathfrak{I} \to \mathcal{B}$, which is generally stressed in $(\mathcal{B}, \Go)$ but whose push-forward $\Fa_* \gamma$ is locally stress-free. Its squared arc-length element in the natural configuration reads $\llangle \Fa \gamma'(t), \Fa \gamma'(t) \rrangle_{\Go} = \llangle \gamma'(t), \gamma'(t) \rrangle_{\Fa^* \Go}$\textemdash see \citep{Yavari2021Eshelby}.}
%---------------------
\begin{equation}
\G = \Fa^* \Go\,, \quad \text{i.e.} \quad \cG_{AB} = \cFa^I{}_A\, \cGo_{IJ}\, \cFa^J{}_B\,.
\end{equation}
%---------------------
The $\G$-inner product on $T_X\mathcal{B}$ is denoted in a similar manner, the material volume form is denoted $dV = {\sqrt{\det \G}\, dX^1 \wedge dX^2 \wedge dX^3}$, and $\nabla$ (with Christoffel symbols ${\Gamma^A}{}_{BC}$) denotes the Levi-Civita connection of $(\mathcal{B}, \G)$. The Jacobian $J$ is defined by requiring that it relate the stress-free (material) and spatial volume elements through $\varphi^* dv = J\, dV$; one finds
%---------------------
\begin{equation}
\label{eq:Jacobian}
J = \sqrt{\frac{\det\g}{\det\G}}\, \det\F = \sqrt{\frac{\det\g}{\det\Go}}\, \det\Fe\,.
\end{equation}
%---------------------
Finally, several sources of eigenstrain may coexist within a solid, each contributing a distinct distortion field, so that the total anelastic distortion factorizes as $\Fa = \prod_{j=1}^{N} \accentset{j}{\mathbf{F}} = \accentset{1}{\mathbf{F}} \hdots \accentset{N}{\mathbf{F}}$.
%--------------------------------------------------------------
\paragraph{Material configuration in nonlinear elasticity.}
Nonlinear elasticity is the eigenstrain-free case, and it is recovered by setting ${\Fa = \mathbf{id}_{T_X\mathcal{B}}}$---the identity map on $T_X\mathcal{B}$, such that the elastic distortion carries the entire deformation, $\Fe = \F$. The material metric then degenerates to the induced Euclidean one, $\G = \Fa^*\Go = \Go$, and the material manifold collapses onto the Euclidean reference configuration $(\mathcal{B}, \Go)$: the natural, stress-free state is globally embeddable in the ambient space, and the body relaxes without residual stress. Accordingly, $dV = d\mathring{V}$, $\nabla = \mathring\nabla$, and the Jacobian~\eqref{eq:Jacobian} reduces to $J = \mathring{J} \coloneq \sqrt{\det\g/\det\Go}\, \det\F$. Throughout, we retain the notation $(\mathcal{B}, \Go)$ for the Euclidean reference configuration and $(\mathcal{B}, \G)$ for the material one, the two coinciding precisely in this elastic case.
%--------------------------------------------------------------
%--------------------------------------------------------------
\paragraph{Musical isomorphisms.}
Recall the musical isomorphisms introduced in \S\ref{S:Phys_D} as metric-dependent instruments.
Therefore, in the framework of anelasticity, where the reference body carries two metrics, $\Go$ and $\G$, one has two such pairs of musical isomorphisms, which we label $\flato, \sharpo$ and $\flat, \sharp$, respectively. In particular, one has $\Go^\sharpo=\Go^{-1}$ and $\G^\sharp=\G^{-1}$.
The ambient space, by contrast, carries only $\g$, so a single pair suffices there and we write it simply as $\flat$ and $\sharp$, with no risk of ambiguity. In particular, one writes $\g^\sharp=\g^{-1}$.
%--------------------------------------------------------------
\paragraph{Adjoint.}
The adjoint $\F^\star$ of the deformation gradient is the unique linear map characterized by the duality identity ${\langle \boldsymbol{\alpha}, \F\mathbf{U} \rangle = \langle \F^\star \boldsymbol{\alpha}, \mathbf{U} \rangle}$, holding for all $\mathbf{U} \in T_X\mathcal{B}$ and all $\boldsymbol{\alpha} \in T^*_{\varphi(X)}\mathcal{S}$. Since it is built purely from the dual pairing, no metric is needed; in components $(\mathrm{F}^\star)^A{}_a = \mathrm{F}^a{}_A$, that is, $\F^\star \circ \varphi(X) = \cF^a{}_A(X)\, dX^A \otimes \partial_a$.
%--------------------------------------------------------------
\paragraph{Transpose.}
Unlike the adjoint, the transpose of $\F$ is metric-dependent, and in the reference configuration the choice between $\Go$ and $\G$ matters. Pairing through $\g$ on the spatial side and $\Go$ on the material side, the transpose $\F^{\mathring{\mathsf{T}}}$ is defined by
$\llangle \F \mathbf{U}, \mathbf{u} \rrangle_{\g} = \llangle \mathbf{U}, \F^{\mathring{\mathsf{T}}} \mathbf{u} \rrangle_{\Go}$,
which yields $\F^{\mathring{\mathsf{T}}} = \Go^\sharpo \F^\star \g$, or $(\mathrm{F}^{\mathring{\mathsf{T}}})^A{}_a = \cGo^{AB} \cF^b{}_B \cg_{ba}$ in components. Repeating the construction with $\G$ in place of $\Go$,
$\llangle \F \mathbf{U}, \mathbf{u} \rrangle_{\g} = \llangle \mathbf{U}, \F^{\mathsf{T}} \mathbf{u} \rrangle_{\G}$,
gives instead $\F^{\mathsf{T}} = \G^\sharp \F^\star \g$, with components $(\mathrm{F}^{\mathsf{T}})^A{}_a = \cG^{AB} \cF^b{}_B \cg_{ba}$.
%--------------------------------------------------------------
\paragraph{Derived strain measures.}
Building on the constructions above, several standard (total) strain measures may now be assembled
%---------------------
\begin{itemize}[topsep=3pt, itemsep=0pt, leftmargin=15pt]
 \item {Piola deformation tensor}: $ \mathbf{B} = \varphi^* \g^\sharp = \F^{-1} \g^\sharp \F^{-\star} $, with components $ \mathrm B^{AB} = (\cF^{-1})^A{}_a \cg^{ab} (\cF^{-1})^B{}_b$;
 \item {right Cauchy--Green tensor}: $ \C = \varphi^* \g = \F^\star \g \F $, with components $ \mathrm C_{AB} = \cF^a{}_A \cg_{ab} \cF^b{}_B $;
 \item {left Cauchy--Green (Finger) tensor}: $ \fb = \varphi_* \G^\sharp = \F \G^\sharp \F^\star $, with components $ \mathrm b^{ab} = \cF^a{}_A \cG^{AB} \cF^b{}_B $;
 \item {inverse Finger tensor}: $ \ic = \varphi_* \G = \F^{-\star} \G \F^{-1} $, with components $ \mathrm c_{ab} = (\cF^{-1})^A{}_a \cG_{AB} (\cF^{-1})^B{}_b $.
\end{itemize}
%---------------------
In anelasticity, one may additionally define the following derived measure:
%---------------------
\begin{itemize}[topsep=3pt, itemsep=0pt, leftmargin=15pt]
 \item {elastic Piola deformation tensor}: $ \Be = \Fe^* \g^\sharp = \Fe^{-1} \g^\sharp \Fe^{-\star} $, with components $ \cBe^{AB} = (\cFe^{-1})^A{}_a \cg^{ab} (\cFe^{-1})^B{}_b $;
 \item {elastic right Cauchy--Green tensor}: $ \Ce = \Fe^* \g = \Fe^\star \g \Fe $, with components $ \cCe_{AB} = \cFe^a{}_A \cg_{ab} \cFe^b{}_B $;
 \item {elastic left Cauchy--Green (Finger) tensor}: $ \be = \Fe_* \Go^\sharpo = \Fe \Go^\sharpo \Fe^\star $, with components $ \cbe^{ab} = \cFe^a{}_A \cGo^{AB} \cFe^b{}_B $;
 \item {elastic inverse Finger tensor}: $ \ce = \Fe_* \Go = \Fe^{-\star} \Go \Fe^{-1} $, with components $ \cce_{ab} = (\cFe^{-1})^A{}_a \cGo_{AB} (\cFe^{-1})^B{}_b $.
\end{itemize}
%---------------------
%-----------------------------
\begin{remark}
\label{rmrk:StrainTransport}
Note that since $\mathbf{G}=\Fa^*\Go$, it follows that $\mathbf{B}=\Fa^*\Be$ and $\mathbf{C}=\Fa^*\Ce$, while $\fb=\be$ and $\ic=\ce$.
\end{remark}
%-----------------------------
%--------------------------------------------------------------
\paragraph{Constitutive assumptions.}
A solid is hyper-anelastic when its stress derives from a scalar energy function \citep{SaYa2025GenColemanNoll}. In anelasticity the eigenstrains encoded in $\Fa$ are stress-relieving, so that the natural configuration $(\mathcal{B}, \G)$ stores no energy and only the recoverable part of the deformation does: the specific free energy\textemdash measured per unit mass\textemdash is accordingly a function of the elastic distortion alone,
$\mathcal W = \mathring{\mathcal W}(X, \Theta, \Fe, \Go, \g\circ\varphi)$,
with $\Theta = \Theta(X,t)$ the temperature field and $\Fe$ reckoned against the Euclidean reference metric $\Go$, consistently with $\Fe: (T_X\mathcal{B}, \Go) \to (T_x\mathcal{S}, \g)$.\footnote{The specific free energy is the Legendre transform $\mathcal W = \mathcal E - \Theta\,\mathcal N$ of the specific internal energy $\mathcal E$ with respect to the conjugate pair $(\Theta, \mathcal N)$, $\mathcal N$ denoting the specific entropy, whence $\mathcal N = -\partial\mathcal W/\partial\Theta$; see \citep{SaYa2025GenColemanNoll}.} Frame indifference (objectivity) constrains this dependence, allowing the energy to be re-expressed through the elastic right Cauchy--Green tensor $\Ce$ as $\mathring{\mathcal W} = \hat{\mathring{\mathcal W}}(X, \Theta, \Ce, \Go)$. Nonlinear elasticity is the degenerate case $\Fa = \mathbf{id}_{T_X\mathcal{B}}$, in which the whole of the deformation is recoverable: $\Fe = \F$ and $\Ce = \C$, and energy becomes a function of the total deformation gradient, $\mathcal W = \mathring{\mathcal W}(X, \Theta, \F, \Go, \g\circ\varphi) = \hat{\mathring{\mathcal W}}(X, \Theta, \C, \Go)$\textemdash the classical hyperelastic free energy referred to the Euclidean reference configuration.
%--------------------------------------------------------------
\paragraph{Stress measures.}
The Cauchy stress follows from the free energy through the Doyle--Ericksen formula \citep{DoyleEricksen1956, SaYa2025GenColemanNoll}, which admits the equivalent representations\footnote{In the sense of definition \eqref{eq:(r,s)_tensor}: $\boldsymbol\sigma$ is a tensor of type $\binom{2}{0}$ and it said to be contravariant of order 2. However, by musical isomorphisms, one is able to construct all the other isometrically equivalent forms.}
%---------------------
\begin{equation} \label{eq:Cauchy-Anelast}
\boldsymbol\sigma
= \varrho\,\g^\sharp \frac{\partial \mathring{\mathcal W}}{\partial \Fe}\Fe^\star
= 2\varrho\, \Fe \frac{\partial \hat{\mathring{\mathcal W}}}{\partial \Ce} \Fe^\star\,,
\end{equation}
%---------------------
where $\varrho$ is the spatial mass density, i.e.\ the mass per unit deformed volume $dv$, related to the material mass density $\rho$\textemdash the mass per unit stress-free volume $dV$\textemdash by $\varrho = \rho/J$. The Cauchy stress is a spatial tensor, contravariant on both legs. Transporting one or both of its legs back to the reference configuration by means of the Piola transform yields the two Piola--Kirchhoff stresses. The first Piola--Kirchhoff stress $\mathbf{P}$ is the two-point tensor of type $\binom{1\ 1}{0\ 0}$ obtained by pulling back the second leg of $\boldsymbol\sigma$,
%---------------------
\begin{equation} \label{eq:FirstPK}
\mathbf{P} = J\,\boldsymbol\sigma\,\F^{-\star} = \rho\,\g^\sharp \frac{\partial \mathring{\mathcal W}}{\partial \Fe}\Fa^{-\star}\,, \quad \text{i.e.} \quad \mathrm{P}^{aA} = J\, \sigma^{ab}\, (\cF^{-1})^A{}_b\,,
\end{equation}
%---------------------
while the second Piola--Kirchhoff stress $\mathbf{S}$ is the fully referential tensor, contravariant of order two, obtained by pulling back both legs,
%---------------------
\begin{equation} \label{eq:SecondPK}
\mathbf{S} = \F^{-1}\mathbf{P} = J\, \F^{-1}\boldsymbol\sigma\,\F^{-\star}\,, \quad \text{i.e.} \quad \mathrm{S}^{AB} = J\, (\cF^{-1})^A{}_a\, \sigma^{ab}\, (\cF^{-1})^B{}_b\,,
\end{equation}
%---------------------
so that $\boldsymbol\sigma = J^{-1}\mathbf{P}\F^\star = J^{-1}\F\mathbf{S}\F^\star$ and $\mathbf{P} = \F\mathbf{S}$. Unlike $\boldsymbol\sigma$, both Piola--Kirchhoff stresses are measured per unit stress-free volume $dV$; correspondingly, the traction $\mathbf{T} = \mathbf{P}\mathbf{N}^\flat$ is reckoned per unit stress-free area, with $\mathbf{N}$ the $\G$-unit normal. Since $\F: (T_X\mathcal{B}, \G) \to (T_x\mathcal{S}, \g)$---see Fig.~\ref{Fig:distortions}, the referential legs of $\mathbf{P}$ and $\mathbf{S}$ are attached to the material configuration $(\mathcal{B}, \G)$ and their spatial legs to $(\mathcal{S}, \g)$. In the elastic limit $\Fa = \mathbf{id}_{T_X\mathcal{B}}$ the referential legs are attached instead to the Euclidean reference configuration, recovering the classical Piola--Kirchhoff stresses of nonlinear elasticity.
%--------------------------------------------------------------
\paragraph{Material symmetry.}
Relative to the Euclidean reference $(\mathcal{B}, \Go)$, the material symmetry group $\mathring{\mathcal G}_X$ of the solid at a point $X \in \mathcal{B}$ collects those $\mathring{\mathbf{K}} \in \mathrm{Orth}(\Go) = \{ \mathbf{Q}: T_X\mathcal{B} \to T_X\mathcal{B} ~|~ \mathbf{Q}^\star \Go \mathbf{Q} = \Go \}$ under which $\mathring{\mathcal W}$ is left unchanged, i.e.
%--------------------------------------------------------------
\begin{equation} \label{Elasticity-Sym-Group-Reform}
\mathring{\mathcal W}(X, \Theta, \mathring{\mathbf{K}}^* \Fe, \mathring{\mathbf{K}}^* \Go, \g)
= \mathring{\mathcal W}(X, \Theta, \Fe, \Go, \g)\,,
\end{equation}
%--------------------------------------------------------------
for an arbitrary elastic distortion $\Fe$.
The symmetries gathered in $\mathring{\mathcal{G}}_X$ can be represented by a family of structural tensors $\mathring{\boldsymbol{\Lambda}}$ that span the tensors left invariant by the group action.\footnote{A finite collection of structural tensors $\mathring{\boldsymbol{\Lambda}}_i$, $i = 1, \dots, N$, suffices to characterize the subgroup $\mathring{\mathcal{G}}_X \subseteq \mathrm{Orth}(\Go)$; see, e.g., \citep{liu1982, boehler1987, zheng1993, zheng1994theory, lu2000covariant, MazzucatoRachele2006}.} Promoting these tensors to independent arguments of the energy renders the constitutive law materially covariant \citep{Lu2012MatCov}: for every linear isomorphism $\boldsymbol{\mathsf{T}}: T_X\mathcal{B} \to T_X\mathcal{B}$ one then has
$\mathring{\mathcal W}(X, \Theta, \Fe, \mathring{\boldsymbol{\Lambda}}, \Go, \g)
= \mathring{\mathcal W}(X, \Theta, \boldsymbol{\mathsf{T}}^* \Fe, \boldsymbol{\mathsf{T}}^* \mathring{\boldsymbol{\Lambda}}, \boldsymbol{\mathsf{T}}^* \Go, \g)$.
Selecting the particular isomorphism $\boldsymbol{\mathsf{T}} = \Fa$ and invoking the identities $\Fa^* \Fe = \Fe \Fa = \F$ together with $\G = \Fa^* \Go = \Fa^\star \Go \Fa$ recasts the energy as
%--------------------------------------------------------------
\begin{equation} \label{Eq:FreeAnelast-Reform}
\mathcal{W}
= \mathcal{W}(X, \Theta, \F, \boldsymbol{\Lambda}, \G, \g)\,,
\end{equation}
%--------------------------------------------------------------
where $\boldsymbol{\Lambda} \coloneq \Fa^* \mathring{\boldsymbol{\Lambda}}$ denotes the structural tensors transported to the configuration $(\mathcal{B}, \G)$. Equation~\eqref{Eq:FreeAnelast-Reform} makes transparent the way in which the material metric $\G = \Fa^* \Go$ emerges of its own accord from the hyper-anelastic description. A further appeal to objectivity permits yet another equivalent form, $\mathcal W = \hat{\mathcal W}(X, \Theta, \C, \boldsymbol{\Lambda}, \G)$. The stress measures of the preceding paragraph may hence be represented equivalently in terms of the total deformation,
%--------------------------------------------------------------
\begin{equation} \label{eq:Stress-Total}
\boldsymbol\sigma
= \varrho\,\g^\sharp \frac{\partial {\mathcal W}}{\partial \F}\F^\star
= 2\varrho\, \F \frac{\partial \hat{{\mathcal W}}}{\partial \C} \F^\star\,, \quad
\mathbf{P} = \rho\,\g^\sharp \frac{\partial {\mathcal W}}{\partial \F}\,, \quad
\mathbf{S} = 2\rho\, \frac{\partial \hat{{\mathcal W}}}{\partial \C}\,,
\end{equation}
%--------------------------------------------------------------
the passage to~\eqref{eq:Stress-Total} from~\eqref{eq:Cauchy-Anelast}--\eqref{eq:SecondPK} being effected by the chain rule $\partial\mathcal{W}/\partial\F = (\partial\mathring{\mathcal W}/\partial\Fe)\,\Fa^{-\star}$. The nonlinear elasticity case follows for $\Fa = \mathbf{id}_{T_X\mathcal{B}}$: the transport is trivial, $\boldsymbol{\Lambda} = \mathring{\boldsymbol{\Lambda}}$ and $\G = \Go$, whence~\eqref{Eq:FreeAnelast-Reform} reduces to~\eqref{Elasticity-Sym-Group-Reform}. The material metric carries no information beyond the embedding geometry, and one recovers the classical notion of material symmetry referred to $(\mathcal{B}, \Go)$.
%--------------------------------------------------------------
%-----------------------------
\subsection{Physical components in nonlinear elasticity}
\label{Subsection:Phys-Elasticity}
Before taking up the anelastic setting, we discuss physical components in nonlinear elasticity, which we recall is the case $\Fa=\mathbf{id}_{T_X\mathcal{B}}$: the material metric coincides with the induced Euclidean metric, $\G=\Go$, and each configuration carries a single metric---$\Go$ on the reference body and $\g$ in the ambient space. Taking $\mathcal{M}=\mathcal{B}$, $\mathcal{N}=\mathcal{S}$, $\M=\Go$, and $\m=\g$, we let $\Lo$ and $\boldsymbol\ell$ denote the orthonormalization maps of~\eqref{eq:2pt-frames} on the reference and spatial legs, respectively. 
Since each leg is equipped with a single metric, the only freedom in the orthonormalization is the choice of frame described by~\eqref{eq:2pt-comp-change}, and orthogonality of $\{X^A\}$ and $\{x^a\}$ means, unambiguously, $\Go$- and $\g$-orthogonality.
%---------------------------------
\paragraph{Deformation gradient.}
Physical components of the deformation gradient $\F:(T_X\mathcal{B},\Go)\to(T_x\mathcal{S},\g)$ follow from~\eqref{eq:Phys-Comp-2pt-prototype} as
%---------------------------------
\begin{equation}
\label{eq:F-phys-elast}
	\hat{\cF}^a{}_A = (\ell^{-1})_b{}^a\,\cF^b{}_B\,\cLo_A{}^B\,,
\end{equation}
%---------------------------------
which in orthogonal coordinates reduces to
%---------------------------------
\begin{equation}
\label{eq:F-phys-elast-orth}
	\hat{\cF}^a{}_A = \sqrt{\frac{\cg_{aa}}{\cGo_{AA}}}\,\cF^a{}_A
	\quad (\text{no summation})\,,
\end{equation}
%---------------------------------
recovering the classical formula of the orthogonal-coordinate literature \citep{Truesdell1953physical,GreenZerna1954,Aris1962}. The components $\hat{\cF}^a{}_A$ are indeed dimensionless even if the coordinates carry mixed dimensions, e.g., in cylindrical coordinates.
%---------------------------------
\paragraph{Jacobian.}
Physical components also clear a familiar abuse of notation. It is common in the continuum
mechanics literature to write $J=\det\mathbf F$, which if taken literally is not well-defined:
a two-point tensor has no determinant independently of a choice of bases. 
Once bases are chosen, $\mathbf F=\cF^a{}_A\,\frac{\partial}{\partial x^a}\otimes dX^A$ is represented by the matrix $\llbracket\cF^a{}_A\rrbracket$, whose determinant is well defined. This matrix determinant is not, in general, equal to the Jacobian; from~\eqref{eq:Jacobian}, they differ by the metric factor $\sqrt{\det\mathbf g/\det\mathbf G}$, which equals unity only when both coordinate bases are orthonormal.
Referring the two legs of $\mathbf F$ to orthonormal frames absorbs precisely
this factor into the frames, so that the classical identity is restored---as it stands---by
physical components. This is the content of the following proposition, which requires
of the frames only orthonormality and positive orientation.
%-----------------------------
\begin{prop}
\label{prop:Jacobian}
Let $\varphi:\mathcal{B}\to\mathcal{S}$ be a deformation between two Riemannian manifolds $(\mathcal{B},\mathbf G)$ and $(\mathcal{S},\mathbf g)$, and let
%---------------------------------
\begin{equation}
	\mathbf F	= \cF^a{}_A\frac{\partial}{\partial x^a}\otimes dX^A\,,\quad 
	\cF^a{}_A=\frac{\partial \varphi^a}{\partial X^A}\,,
\end{equation}
%---------------------------------
be its deformation gradient in local coordinates $\{X^A\}$ on $\mathcal{B}$ and $\{x^a\}$ on $\mathcal{S}$. Let $\{\hat{\mathbf E}_A\}$ and $\{\hat{\mathbf e}_a\}$ be positively-oriented orthonormal frames related to the coordinate bases by
%---------------------------------
\begin{equation} \label{Frame-Transformation-Prop}
	\hat{\mathbf E}_A=\cL_A{}^B\frac{\partial}{\partial X^B}\,,\quad
	\hat{\mathbf e}_a=\ell_a{}^b\frac{\partial}{\partial x^b}\,,
\end{equation}
%---------------------------------
and let $\{\hat{\boldsymbol\Omega}^A\}$ and $\{\hat{\boldsymbol\omega}^a\}$ be their corresponding dual coframes, respectively.
If $\mathbf F=\hat{\cF}^a{}_A\,\hat{\mathbf e}_a\otimes \hat{\boldsymbol\Omega}^A$ denotes the representation of $\mathbf F$ in these orthonormal frames, then
%---------------------------------
\begin{equation} \label{J-FHat}
	J=\det\llbracket \hat{\cF}^a{}_A \rrbracket\,,
\end{equation}
%---------------------------------
where $J=	\det\llbracket {\cF} \rrbracket \sqrt{{\det\mathbf g}/{\det\mathbf G}}$ is the Jacobian of the deformation.
\end{prop}
\begin{proof}
We know that $\hat{\cF}^a{}_A=\big(\ell^{-1}\big)_b{}^a\,\cF^b{}_B\,\cL_A{}^B$. Therefore, taking determinants gives
%---------------------------------
\begin{equation} \label{Det-Hat-F}
	\det\llbracket \hat{\cF}^a{}_A \rrbracket
	=\det\llbracket \big(\ell^{-1}\big)_b{}^a\rrbracket \det\llbracket {\cF} \rrbracket \det\llbracket \cL_A{}^B \rrbracket\,.
\end{equation}
%---------------------------------
We now compute the determinants of the orthonormalization matrices. Since $\{\hat{\mathbf e}_a\}$ is orthonormal,
$\delta_{ab} = \llangle \hat{\mathbf e}_a,\hat{\mathbf e}_b\rrangle_{\mathbf g}
= \ell_a{}^c\,\ell_b{}^d\,\cg_{cd}$.
Taking determinants, one obtains $(\det \boldsymbol{\ell})^2 \det \mathbf g=1$.
Because the spatial orthonormal frame is assumed to be positively oriented,
$\det \boldsymbol{\ell}^{-1}=1/\det \boldsymbol{\ell}=\sqrt{\det \mathbf g}$.
Similarly, the orthonormality of $\{\hat{\mathbf E}_A\}$ implies that
$\delta_{AB}= \llangle \hat{\mathbf E}_A,\hat{\mathbf E}_B \rrangle_{\mathbf G}
= \cL_A{}^C\,\cL_B{}^D\,\cG_{CD}$,
and hence $(\det\Lt)^2 \det \mathbf G=1$.
Since the material orthonormal frame is positively oriented, we have
$\det \Lt=1/\sqrt{\det \mathbf G}$.
Substituting these identities into \eqref{Det-Hat-F} yields
%---------------------------------
\begin{equation}
	\det\llbracket \hat{\cF}^a{}_A \rrbracket=\sqrt{\det \mathbf g}\,\det\llbracket {\cF} \rrbracket\,\frac{1}{\sqrt{\det \mathbf G}} =J\,.
\end{equation}
%---------------------------------
\end{proof}
%-----------------------------
%-----------------------------
\begin{remark}
We have shown that the familiar identity $J=\det\mathbf F$ is valid precisely when $\mathbf F$ is
represented by its physical components. Cartesian coordinates are the special case in which
the coordinate bases are already orthonormal, so that coordinate and physical components
coincide and the identity may be read literally. A proof of $J=\det\hat{\mathbf F}$ for
orthogonal curvilinear coordinates was given in \citep{PradhanYavari2024PhaseChange};
the above proposition establishes the stronger result that it holds for arbitrary
coordinates and arbitrary positively-oriented orthonormal frames.
\end{remark}
%-----------------------------
%---------------------------------
\paragraph{Derived strain measures.}
The referential measures $\C$ and $\mathbf{B}$ are orthonormalized by $\Lo$ and the spatial measures $\ic$ and $\fb$ by $\boldsymbol\ell$, giving the (dimensionless) physical components
%---------------------------------
\begin{equation}
\label{eq:strain-phys-elast}
\begin{aligned}
	&\hat{\cC}_{AB} = \cLo_A{}^C\,\cLo_B{}^D\,\cC_{CD}\,, &\quad
	&\hat{\mathrm B}^{AB} = (\cLo^{-1})_C{}^A\,(\cLo^{-1})_D{}^B\,\mathrm B^{CD}\,, \\
	&\hat{\mathrm c}_{ab} = \ell_a{}^c\,\ell_b{}^d\,\mathrm c_{cd}\,, &\quad
	&\hat{\mathrm b}^{ab} = (\ell^{-1})_c{}^a\,(\ell^{-1})_d{}^b\,\mathrm b^{cd}\,,
\end{aligned}
\end{equation}
%---------------------------------
which in orthogonal coordinates reduce to
%---------------------------------
\begin{equation}
\label{eq:strain-phys-elast-orth}
\begin{aligned}
	&\hat{\cC}_{AB} = \frac{\cC_{AB}}{\sqrt{\cGo_{AA}\,\cGo_{BB}}}\,, &\quad
	&\hat{\mathrm B}^{AB} = \sqrt{\cGo_{AA}\,\cGo_{BB}}\;\mathrm B^{AB}\,, \\
	&\hat{\mathrm c}_{ab} = \frac{\mathrm c_{ab}}{\sqrt{\cg_{aa}\,\cg_{bb}}}\,, &\quad
	&\hat{\mathrm b}^{ab} = \sqrt{\cg_{aa}\,\cg_{bb}}\;\mathrm b^{ab}\,,
\end{aligned}
\end{equation}
%---------------------------------
with no summation on repeated indices. Substituting~\eqref{eq:F-phys-elast} into the definitions $\C=\F^\star\g\F$, $\mathbf{B}=\F^{-1}\g^\sharp\F^{-\star}$, $\ic=\F^{-\star}\Go\F^{-1}$, and $\fb=\F\Go^\sharpo\F^\star$, the metrics are absorbed into the frames and one finds
%---------------------------------
\begin{equation}
\label{eq:strain-matrix-identities}
\begin{aligned}
	&\hat{\cC}_{AB}=\hat{\cF}^a{}_A\,\hat{\cF}^b{}_B\,\delta_{ab}\,, &\quad
	&\hat{\mathrm B}^{AB} = (\hat{\cF}^{-1})^A{}_a\,(\hat{\cF}^{-1})^B{}_b\,\delta^{ab}\,, \\
	&\hat{\mathrm c}_{ab} = (\hat{\cF}^{-1})^A{}_a\,(\hat{\cF}^{-1})^B{}_b\,\delta_{AB}\,, &\quad
	&\hat{\mathrm b}^{ab} =\hat{\cF}^a{}_A\,\hat{\cF}^b{}_B\,\delta^{AB}\,,
\end{aligned}
\end{equation}
%---------------------------------
which correspond to the classical matrix identities of Cartesian continuum mechanics, reproduced verbatim.
%---------------------------------
\paragraph{Stress measures.}
The Cauchy stress $\boldsymbol\sigma$~\eqref{eq:Cauchy-Anelast} is spatial and contravariant on both legs; the first Piola--Kirchhoff stress $\mathbf{P}$~\eqref{eq:FirstPK} is contravariant on one leg of each kind; and the second Piola--Kirchhoff stress $\mathbf{S}$~\eqref{eq:SecondPK} is referential and contravariant on both legs. Here $\Fa = \mathbf{id}_{T_X\mathcal{B}}$, so no ambiguity attends the reference metric and both referential legs are orthonormalized against $\Go$. Their physical components read
%---------------------------------
\begin{equation}
\label{eq:stress-phys-elast}
	\hat{\sigma}^{ab} = (\ell^{-1})_c{}^a\,(\ell^{-1})_d{}^b\,\sigma^{cd}\,, \quad
	\hat{\cP}^{aA} = (\ell^{-1})_b{}^a\,(\cLo^{-1})_B{}^A\,\cP^{bB}\,, \quad
	\hat{\mathrm S}^{AB} = (\cLo^{-1})_C{}^A\,(\cLo^{-1})_D{}^B\,\mathrm S^{CD}\,,
\end{equation}
%---------------------------------
reducing in orthogonal coordinates to
%---------------------------------
\begin{equation}
\label{eq:stress-phys-elast-orth}
	\hat{\sigma}^{ab} = \sqrt{\cg_{aa}\,\cg_{bb}}\;\sigma^{ab}\,, \quad
	\hat{\cP}^{aA} = \sqrt{\cg_{aa}\,\cGo_{AA}}\;\cP^{aA}\,, \quad
	\hat{\mathrm S}^{AB} = \sqrt{\cGo_{AA}\,\cGo_{BB}}\;\mathrm S^{AB}
	\quad (\text{no summation})\,.
\end{equation}
%---------------------------------
Unlike the coordinate components, whose physical dimensions vary from entry to entry in curvilinear coordinates, every entry of~\eqref{eq:stress-phys-elast-orth} carries the dimension of stress.
%---------------------------------
%---------------------------------
\subsection{Physical components in finite anelasticity}
\label{Subsection:Phys-Anelasticity}
We now specialize the physical components construction to finite anelasticity, taking $\mathcal{M}=\mathcal{B}$, $\mathcal{N}=\mathcal{S}$, and $\phi=\varphi$. Physical components require an orthonormal frame on each leg, and hence a metric on each tangent space to which a leg is attached. On the spatial side this poses no difficulty, the ambient space equipped with the canonical metric $\g$, i.e. $\m=\g$. The reference body, by contrast, carries two: the induced Euclidean metric $\Go=\iota^*\g$ and the material (stress-free) metric $\G=\Fa^*\Go$. It is the BKL split $\F=\Fe\Fa$ that forces the distinction to be drawn, interposing the material configuration $(\mathcal{B},\G)$ between $(\mathcal{B},\Go)$ and $(\mathcal{S},\g)$ and thereby distributing the two reference metrics among its factors: $\Fe$ carries $\Go$ on its reference leg, $\F$ carries $\G$, and $\Fa$---being the map between the two referential pictures---one of each; see Fig.~\ref{Fig:distortions}. Each referential leg is accordingly orthonormalized against the metric with respect to which the tensor is reckoned, a choice dictated by the kinematic and constitutive role of that tensor rather than by convenience. We start with general curvilinear coordinate charts $\{X^A\}$ on $\mathcal{B}$ and $\{x^a\}$ on $\mathcal{S}$, no orthogonality being assumed with respect to any of the three metrics, and we proceed to examine physical components of key quantities in finite anelasticity.
%---------------------------------
\paragraph{The deformation gradient.}
For the deformation gradient and each of its two distortions, we record the metric that its reference leg calls for and the resulting physical components. Let $\Lo$ and $\Lt$ denote the reference-leg orthonormalization maps~\eqref{eq:2pt-frames} associated with $\Go$ and $\G$, respectively, and let $\boldsymbol\ell$ denote the spatial-leg orthonormalization map associated with $\g$.
%-----------------------------
\begin{itemize}[topsep=3pt, itemsep=0pt, leftmargin=15pt]
%---------------------------------
\item \emph{Elastic distortion $\Fe:(T_X\mathcal{B},\Go)\to(T_x\mathcal{S},\g)$.} The elastic distortion is attached to the Euclidean reference configuration, consistent with the constitutive representation $\mathcal{W}=\mathring{\mathcal{W}}(X,\Theta,\Fe,\mathring{\boldsymbol{\Lambda}},\Go,\g)$, in which $\Fe$ is measured against $\Go$. Its physical components are taken with respect to $\Go$ on the reference leg and $\g$ on the spatial leg,
%---------------------------------
\begin{equation}
\label{eq:Fe-phys-general}
	\hat{\cFe}^a{}_A = (\ell^{-1})_b{}^a\,\cFe^b{}_B\,\cLo_A{}^B\,.
\end{equation}
%---------------------------------
%---------------------------------
\item \emph{Anelastic distortion $\Fa:(T_X\mathcal{B},\G)\to(T_X\mathcal{B},\Go)$.} The anelastic distortion connects the two referential pictures, with the metric $\G$ on its domain and the metric $\Go$ on its codomain by the construction $\G=\Fa^*\Go$. Its physical components are taken with respect to $\G$ on the domain and $\Go$ on the codomain,
%---------------------------------
\begin{equation}
\label{eq:Fa-phys-general}
	\hat{\cFa}^A{}_B = (\cLo^{-1})_C{}^A\,\cFa^C{}_D\,\cL_B{}^D\,,
\end{equation}
%---------------------------------
and the resulting matrix $\llbracket\hat{\cFa}^A{}_B\rrbracket$ is orthogonal.\footnote{This follows directly from the isometry condition $\G=\Fa^*\Go$: expanding $\llangle\Fa\mathbf{U},\Fa\mathbf{V}\rrangle_{\Go}=\llangle\mathbf{U},\mathbf{V}\rrangle_{\G}$ on the $\G$- and $\Go$-orthonormal bases yields $\llbracket\hat{\cFa}\rrbracket^{\top}\,\llbracket\hat{\cFa}\rrbracket=\llbracket\delta^A_B\rrbracket$.}
%---------------------------------
\item \emph{Total deformation gradient $\F:(T_X\mathcal{B},\G)\to(T_x\mathcal{S},\g)$.} The total deformation is attached to the material configuration, consistent with $\mathcal{W}=\mathcal{W}(X,\Theta,\F,\boldsymbol{\Lambda},\G,\g)$, in which it is measured against the natural stress-free reference (material) manifold $(\mathcal{B},\G)$. Its physical components are taken with respect to $\G$ on the reference leg and $\g$ on the spatial leg,
%---------------------------------
\begin{equation}
\label{eq:F-phys-general}
	\hat{\cF}^a{}_A = (\ell^{-1})_b{}^a\,\cF^b{}_B\,\cL_A{}^B\,.
\end{equation}
%---------------------------------
\end{itemize}
%---------------------------------
Consistency with the action of $\F$ on vectors follows directly: writing $\hat{\mathrm u}^a$ and $\hat{\mathrm U}^A$ for physical components of $\mathbf u=\F\mathbf U$ and $\mathbf U$, equations~\eqref{eq:Phys_Comp_gen} and~\eqref{eq:F-phys-general} give
%---------------------------------
\begin{equation}
	\hat{\mathrm u}^a
	=(\ell^{-1})_b{}^a\,\mathrm{u}^b
	=(\ell^{-1})_b{}^a\,\cF^b{}_B\,\mathrm U^B
	=(\ell^{-1})_b{}^a\,\cF^b{}_B\,\cL_A{}^B\,\hat{\mathrm U}^A
	=\hat{\cF}^a{}_A\,\hat{\mathrm U}^A\,.
\end{equation}
%---------------------------------
Moreover, from~\eqref{eq:2pt-comp-change} physical components of each distortion are defined only up to independent orthogonal transformations of its two frames, the spatial leg gauged by $O(T_x\mathcal{S},\g)$ and the reference leg by $O(T_X\mathcal{B},\Go)$ or $O(T_X\mathcal{B},\G)$ as appropriate.
%---------------------------------
\begin{remark}[Consistency with the split]
\label{remark:split-consistency}
When each distortion is represented with respect to the appropriate frames on its two legs, their physical-component representations satisfy the same composition rule as the distortions themselves,
%---------------------------------
\begin{equation}\label{eq:Fhat-composition}
 \hat{\F} = \hat{\Fe}\,\hat{\Fa}\,,
\end{equation}
%---------------------------------
where the intermediate $\Go$-orthonormal frame cancels between the two distortions.
Recalling that $\llbracket\hat{\cFa}^A{}_B\rrbracket$ is an orthogonal matrix ($\Fa:(T_X\mathcal{B},\G)\to(T_X\mathcal{B},\Go)$ being an isometry), its singular values are all unity, and the gauge freedom found in~\eqref{eq:2pt-comp-change} may recast it as $\hat{\cFa}^A{}_B=\delta^A_B$.\footnote{From~\eqref{eq:2pt-comp-change}, a change of the domain and codomain frames sends $\hat{\cFa}^A{}_B\mapsto(\mathrm R^{-1})_C{}^A\,\hat{\cFa}^C{}_D\,\mathrm Q_B{}^D$, with $\mathbf{Q}\in O(T_X\mathcal{B},\G)$ and $\mathbf{R}\in O(T_X\mathcal{B},\Go)$ represented by orthogonal matrices in the respective orthonormal frames. As $\llbracket\hat{\cFa}^A{}_B\rrbracket$ is itself orthogonal, the choices $\llbracket\mathrm Q\rrbracket=\llbracket\hat{\cFa}\rrbracket^{-1}$ with $\mathbf{R}=\mathbf{id}_{T_X\mathcal{B}}$, or $\llbracket\mathrm R\rrbracket=\llbracket\hat{\cFa}\rrbracket$ with $\mathbf{Q}=\mathbf{id}_{T_X\mathcal{B}}$, are both admissible and yield $\llbracket\hat{\cFa}\rrbracket=\mathbf{id}_{T_X\mathcal{B}}$.} The invariant content of $\hat{\Fa}$ is therefore empty, and by~\eqref{eq:Fhat-composition} the representations $\hat{\F}$ and $\hat{\Fe}$ share the same singular values: in physical components, the anelastic distortion representation acts as a simple change of orthonormal frame between the material and the Euclidean reference configurations. This must not be read as saying that the anelastic distortion is physically meaningless; its essence is not lost but merely relocated: $\Fa$ determines the material metric $\G=\Fa^*\Go$, against which the reference leg of $\F$ is orthonormalized, and it is thus carried into $\hat{\F}$ by the orthonormalization map $\Lt$.
\end{remark}
%---------------------------------
%---------------------------------
\paragraph{Derived strain measures.}
The derived strain measures divide according to the configuration their legs are attached to. The total referential measures $\C$ and $\mathbf{B}$ are attached to the material configuration $(\mathcal{B},\G)$ and are hence orthonormalized by $\Lt$, their elastic counterparts $\Ce$ and $\Be$ are defined with respect to the Euclidean reference configuration $(\mathcal{B},\Go)$ and hence are orthonormalized via $\Lo$, while the spatial measures $\ic$ and $\fb$ carry $\g$ on both legs and are indifferent to the distinction. Their physical components read
%---------------------------------
\begin{equation}
\label{eq:strain-phys-anel}
\begin{aligned}
	&\hat{\cC}_{AB} = \cL_A{}^C\,\cL_B{}^D\,\cC_{CD}\,, &\quad
	&\hat{\mathrm B}^{AB} = (\cL^{-1})_C{}^A\,(\cL^{-1})_D{}^B\,\mathrm B^{CD}\,, \\
	&\hat{\cCe}_{AB} = \cLo_A{}^C\,\cLo_B{}^D\,\cCe_{CD}\,, &\quad
	&\hat{\cBe}^{AB} = (\cLo^{-1})_C{}^A\,(\cLo^{-1})_D{}^B\,\cBe^{CD}\,, \\
	&\hat{\mathrm c}_{ab} = \ell_a{}^c\,\ell_b{}^d\,\mathrm c_{cd}\,, &\quad
	&\hat{\mathrm b}^{ab} = (\ell^{-1})_c{}^a\,(\ell^{-1})_d{}^b\,\mathrm b^{cd}\,.
\end{aligned}
\end{equation}
%---------------------------------
Substituting~\eqref{eq:F-phys-general} into the definitions $\C=\F^\star\g\F$, $\mathbf{B}=\F^{-1}\g^\sharp\F^{-\star}$, $\ic=\F^{-\star}\G\F^{-1}$, and $\fb=\F\G^\sharp\F^\star$, the metrics are absorbed into the frames and one finds
%---------------------------------
\begin{equation}
\label{eq:strain-matrix-identities-anel}
\begin{aligned}
	&\hat{\cC}_{AB}=\hat{\cF}^a{}_A\,\hat{\cF}^b{}_B\,\delta_{ab}\,, &\quad
	&\hat{\mathrm B}^{AB} = (\hat{\cF}^{-1})^A{}_a\,(\hat{\cF}^{-1})^B{}_b\,\delta^{ab}\,, \\
	&\hat{\cCe}_{AB}=\hat{\cFe}^a{}_A\,\hat{\cFe}^b{}_B\,\delta_{ab}\,, &\quad
	&\hat{\cBe}^{AB} = (\hat{\cFe}^{-1})^A{}_a\,(\hat{\cFe}^{-1})^B{}_b\,\delta^{ab}\,, \\
	&\hat{\mathrm c}_{ab} = (\hat{\cF}^{-1})^A{}_a\,(\hat{\cF}^{-1})^B{}_b\,\delta_{AB}\,, &\quad
	&\hat{\mathrm b}^{ab} =\hat{\cF}^a{}_A\,\hat{\cF}^b{}_B\,\delta^{AB}\,.
\end{aligned}
\end{equation}
%---------------------------------
%---------------------------------
\begin{remark}[Total and elastic measures in physical components]
\label{remark:total-vs-elastic-phys}
The spatial identities $\fb=\be$ and $\ic=\ce$ are preserved upon passage to physical components. Indeed, by~\eqref{eq:Fhat-composition} and the orthogonality of $\hat{\Fa}$,
$\hat{\F}\hat{\F}^\top=\hat{\Fe}\hat{\Fa}\hat{\Fa}^\top\hat{\Fe}^\top=\hat{\Fe}\hat{\Fe}^\top$.
The total and elastic referential measures remain distinct, but their pullback relations $\C=\Fa^*\Ce$ and $\mathbf{B}=\Fa^*\Be$ take the following form in physical components:
%---------------------------------
\begin{equation}
\label{eq:C-conjugation}
	\hat{\C} = \hat{\Fa}^\top\,\hat{\Ce}\,\hat{\Fa}\,, \quad
	\hat{\mathbf{B}} = \hat{\Fa}^\top\,\hat{\Be}\,\hat{\Fa}\,.
\end{equation}
%---------------------------------
Thus, the total and elastic referential measures are related by orthogonal similarity transformations and therefore have the same invariants. They may be made identical by the gauge choice of Remark~\ref{remark:split-consistency}, for which $\hat{\cFa}^A{}_B=\delta^A_B$.
\end{remark}
%---------------------------------
%---------------------------------
\paragraph{Stress measures.}
The Cauchy stress $\boldsymbol\sigma$ is spatial and contravariant on both legs; the first Piola--Kirchhoff stress $\mathbf{P}$ is contravariant on one leg of each kind; and the second Piola--Kirchhoff stress $\mathbf{S}$ is referential and contravariant on both legs. Both Piola--Kirchhoff stresses are measured per unit stress-free volume $dV$, the traction $\mathbf{T}=\mathbf{P}\mathbf{N}^{\flat}$ being reckoned per unit stress-free area with $\mathbf{N}$ the $\G$-unit normal; their referential legs are therefore attached to the material configuration $(\mathcal{B},\G)$ and orthonormalized by $\Lt$---and not by $\Lo$, as in the elastic case, where $\Fa=\mathbf{id}_{T_X\mathcal{B}}$ collapses the two maps into one. Their physical components read
%---------------------------------
\begin{equation}
\label{eq:stress-phys-general}
	\hat{\sigma}^{ab} = (\ell^{-1})_c{}^a\,(\ell^{-1})_d{}^b\,\sigma^{cd}\,,\quad
	\hat{\cP}^{aA} = (\ell^{-1})_b{}^a\,(\cL^{-1})_B{}^A\,\cP^{bB}\,,\quad
	\hat{\mathrm{S}}^{AB} = (\cL^{-1})_C{}^A\,(\cL^{-1})_D{}^B\,\mathrm{S}^{CD}\,.
\end{equation}
%---------------------------------
Unlike the coordinate components, whose physical dimensions vary from entry to entry in curvilinear coordinates, every entry of~\eqref{eq:stress-phys-general} carries the dimension of stress. The Piola transforms relating the three measures are likewise unencumbered by metric factors once all legs are referred to orthonormal frames, the referential frames cancelling between $\hat{\mathbf{P}}$, $\hat{\mathbf{S}}$, and $\hat{\F}$,
%---------------------------------
\begin{equation}
\label{eq:stress-cartesian-identities}
	\hat{\mathbf{P}} = \hat{\F}\,\hat{\mathbf{S}}\,,\quad
	\hat{\boldsymbol\sigma} = J^{-1}\,\hat{\mathbf{P}}\,\hat{\F}^{\top} = J^{-1}\,\hat{\F}\,\hat{\mathbf{S}}\,\hat{\F}^{\top}\,.
\end{equation}
%---------------------------------
%--------------------------------------------------------------
\paragraph{Physical components in orthogonal coordinates.}
We now turn to orthogonal coordinates---the orthonormalization maps become diagonal and physical components follow from the coordinate ones by~\eqref{eq:Ortho_Coord_M}. Orthogonality, however, is defined in relation to a metric, and the reduction is available on a given leg only when the coordinates are orthogonal with respect to the very metric against which that leg is orthonormalized. On the spatial side this qualification is unambiguous: the ambient space carries the single metric $\g$, and $\g$-orthogonality of $\{x^a\}$ is all that can be meant. On the referential side it is however ambiguous, as $\{X^A\}$ may be orthogonal with respect to one of the two referential metrics and not necessarily the other: since $\G=\Fa^*\Go$, the components $\cG_{AB}=\cFa^I{}_A\,\cGo_{IJ}\,\cFa^J{}_B$ need not be diagonal when $\cGo_{AB}$ is, nor conversely. Orthogonal referential coordinates must therefore name their metric, and no single chart $\{X^A\}$ can, in general, serve both.
Three conditions are thus in play---$\g$-orthogonality of $\{x^a\}$, and $\Go$- and $\G$-orthogonality of $\{X^A\}$---and none implies another.\footnote{The spatial and referential conditions are independent, $\{X^A\}$ and $\{x^a\}$ being coordinate charts on distinct manifolds. They collapse into one another only in special circumstances: if $\{X^A\}$ is the restriction under the inclusion map $\iota:\mathcal{B}\hookrightarrow\mathcal{S}$ of the ambient chart $\{x^a\}$, then $\cGo_{AB}$ and $\cg_{ab}$ are the same functions of that chart and $\g$-orthogonality of the latter is $\Go$-orthogonality of the former. This does not hold for coordinates convected by the motion, for which $\g$-orthogonality of $\{x^a\}$ amounts to diagonality of $\C=\varphi^*\g$ and is a condition on the deformation rather than on the chart alone. The two referential conditions are independent of each other by the non-diagonality argument just given, $\Fa$ being an arbitrary distortion.} Let us consider a case of practical interest: $\{x^a\}$ is $\g$-orthogonal and $\{X^A\}$ is $\Go$-orthogonal---the typical case for a body whose embedded Euclidean geometry is adapted to the coordinates---but $\{X^A\}$ is not necessarily $\G$-orthogonal. Every leg orthonormalized against $\g$ or $\Go$ contributes a diagonal factor, as given by~\eqref{eq:Ortho_Coord_M}, while every leg orthonormalized against $\G$ retains the generally non-diagonal $\Lt$, as given by~\eqref{eq:beta-explicit}.
%-----------------------------
\begin{itemize}[topsep=\smallskipamount, itemsep=\smallskipamount, parsep=0pt, leftmargin=0pt, itemindent=\parindent, labelwidth=\parindent, labelsep=0pt, align=left]
%---------------------------------
\item The deformation gradient and its split reduce accordingly,
%---------------------------------
\begin{equation}\label{eq:phys-orth-anel}
	\hat{\cFe}^a{}_A = \sqrt{\frac{\cg_{aa}}{\cGo_{AA}}}\;\cFe^a{}_A\,, \quad
	\hat{\cFa}^A{}_B = \sqrt{\cGo_{AA}}\;\cFa^A{}_D\,\cL_B{}^D\,, \quad
	\hat{\cF}^a{}_A = \sqrt{\cg_{aa}}\;\cF^a{}_B\,\cL_A{}^B\,,
\end{equation}
%---------------------------------
with no summation on $a$ and $A$. Only the elastic distortion in \eqref{eq:phys-orth-anel} trivially reduces on both legs following the classical formula of the orthogonal-coordinate literature, now referred to the Euclidean reference configuration. The anelastic distortion reduces on its codomain alone and the total deformation gradient on its spatial leg alone, their common $\G$-orthonormalized referential leg being served by neither hypothesis.
%---------------------------------
\item The derived strain measures divide the same way. Those referred to $\g$ and $\Go$ reduce fully,
%---------------------------------
\begin{subequations}\label{eq:strain-phys-orth-anel}
\begin{alignat}{2}
	\label{eq:spatial-strain-phys-orth-anel}
	&\hat{\mathrm c}_{ab} = \frac{\mathrm c_{ab}}{\sqrt{\cg_{aa}\,\cg_{bb}}}\,, \quad&
	&\hat{\mathrm b}^{ab} = \sqrt{\cg_{aa}\,\cg_{bb}}\;\mathrm b^{ab}\,, \\
	\label{eq:elastic-strain-phys-orth-anel}
	&\hat{\cCe}_{AB} = \frac{\cCe_{AB}}{\sqrt{\cGo_{AA}\,\cGo_{BB}}}\,, \quad&
	&\hat{\cBe}^{AB} = \sqrt{\cGo_{AA}\,\cGo_{BB}}\;\cBe^{AB}\,.
\end{alignat}
\end{subequations}
%---------------------------------
with no summation on repeated indices, while their total referential counterparts, orthonormalized against $\G$, do not reduce at all and retain the general form~\eqref{eq:strain-phys-anel},
%---------------------------------
\begin{equation}\label{eq:total-strain-phys-orth-anel}
	\hat{\cC}_{AB} = \cL_A{}^C\,\cL_B{}^D\,\cC_{CD}\,, \quad
	\hat{\mathrm B}^{AB} = (\cL^{-1})_C{}^A\,(\cL^{-1})_D{}^B\,\mathrm B^{CD}\,.
\end{equation}
%---------------------------------
Under the referential chart one would naturally adopt, then, it is the elastic strain measures that take the classical orthogonal-coordinate form and the total ones that do not---the reverse of nonlinear elasticity, where the two families coincide and both reduce.
%---------------------------------
\item The stress measures divide by leg rather than by family, the Cauchy stress being spatial, the second Piola--Kirchhoff stress referential, and the first Piola--Kirchhoff stress two-point,
%---------------------------------
\begin{equation}\label{eq:stress-phys-orth-anel}
	\hat{\sigma}^{ab} = \sqrt{\cg_{aa}}\,\sigma^{ab}\,\sqrt{\cg_{bb}}\,, \quad
	\hat{\cP}^{aA} = \sqrt{\cg_{aa}}\,\cP^{aB}\,(\cL^{-1})_B{}^A\,, \quad
	\hat{\mathrm S}^{AB} = (\cL^{-1})_C{}^A\,\mathrm S^{CD}\,(\cL^{-1})_D{}^B\,,
\end{equation}
%---------------------------------
with no summation on $a$ and $b$. Only $\boldsymbol\sigma$ reduces fully; $\mathbf{P}$ reduces on its spatial leg alone and $\mathbf{S}$ on neither, both being measured per unit stress-free volume and hence attached to $(\mathcal{B},\G)$.
\end{itemize}
%---------------------------------
%---------------------------------
\begin{remark}[Coordinates orthogonal with respect to both reference metrics]
\label{remark:Fa-diagonal}
Suppose the referential coordinates are $\G$-orthogonal in addition, i.e.
$\cG_{AB}=\llangle\Fa\partial_A,\Fa\partial_B\rrangle_{\Go}=0$, for any $A\neq B$,
so that $\Fa$ carries the coordinate directions to mutually $\Go$-orthogonal ones. Then $\Lt$ is diagonal like $\Lo$ and $\boldsymbol\ell$, and every leg of every quantity reduces. In particular,~\eqref{eq:phys-orth-anel} yields
%---------------------------------
\begin{equation}\label{eq:Fa-phys-both-orth}
	\hat{\cFa}^A{}_B=\sqrt{\frac{\cGo_{AA}}{\cG_{BB}}}\;\cFa^A{}_B
	\quad (\text{no summation})\,,
\end{equation}
%---------------------------------
which is orthogonal, as it must be: it maps the $\G$-orthonormal coordinate frame onto the $\Go$-orthonormal one. The two frames coincide exactly when $\Fa$ is diagonal in $\{X^A\}$, and then $\hat{\Fa}$ is the identity. Indeed, $\cG_{AA}=(\cFa^A{}_A)^2\,\cGo_{AA}$ in that case (no summation on $A$), so that
%---------------------------------
\begin{equation}\label{eq:Fa-phys-identity}
	\hat{\cFa}^A{}_B = \delta^A_B\,,
\end{equation}
%---------------------------------
and $\hat{\F}=\hat{\Fe}$ by~\eqref{eq:Fhat-composition}. The raw components hide this: $\cFa^A{}_B$ rescales each coordinate direction by a nontrivial factor, which the adapted orthonormal frames absorb.
Note that~\eqref{eq:Fa-phys-identity} differs in kind from the trivialization of Remark~\ref{remark:split-consistency}. There the identity is reached by choosing a gauge, and is available for any $\Fa$. Here no gauge remains to be chosen, each metric singling out its normalized coordinate basis; the chart fixes the frames, and $\hat{\Fa}$ is the identity only if $\Fa$ is diagonal. Diagonality is thus the condition for the coordinate-adapted frames to be trivializing ones.
\end{remark}
%---------------------------------
%-----------------------------
\begin{remark}[Uniqueness of the material metric]
\label{Remark:Material-metric}
The BKL split is not determined by the material response: if $\mathring{\mathbf{K}}\in\mathring{\mathcal{G}}_X$, the pair $(\mathring{\mathbf{K}}^{-1}\Fa,\,\Fe\mathring{\mathbf{K}})$ leaves both $\F$ and the energy unchanged by~\eqref{Elasticity-Sym-Group-Reform}, the relaxed configuration being fixed only up to the material symmetry group. The material metric does not inherit this freedom. Since $\mathring{\mathcal{G}}_X\subseteq\mathrm{Orth}(\Go)$,
$(\mathring{\mathbf{K}}^{-1}\Fa)^*\Go=\Fa^*(\mathring{\mathbf{K}}^{-1})^*\Go=\Fa^*\Go=\G$,
and $\G$ is thus invariant: the anelastic distortion is determined only up to material symmetry, while the material metric it induces is uniquely determined. The structure of anelasticity therefore introduces no ambiguity of its own into the reference geometry, and the physical component representations for $\F$ are correspondingly unaffected.
This conclusion relies on the material symmetry group being a subgroup of $\mathrm{Orth}(\Go)$. If the symmetry group contains non-isometric transformations, as in the case of fluids, symmetry-related BKL splits may induce different material metrics.
\end{remark}
%-----------------------------
%---------------------------------
\subsection{Physical components beyond anelasticity}
\label{Sec:Beyond-Anelasticity}
The framework of \S\ref{Section:Physical-Components} is general: it is a construction on Riemannian manifolds, and on two-point tensors over a map between two such manifolds, that nowhere refers to the kinematics of elasticity or anelasticity. Those settings, treated in \S\ref{Subsection:Phys-Elasticity} and \S\ref{Subsection:Phys-Anelasticity}, merely distinguish the configurations and their metrics. The framework may hence extend beyond anelasticity to general inelastic processes, and in what follows we discuss viscoelasticity and visco-anelasticity, referring the reader to \citep{SaYa2024viscoelasticity, SaYa2026visco} for the underlying geometric framework.
%---------------------------------
\paragraph{Viscoelasticity.} Viscoelasticity is the eigenstrain-free inelastic response in which the deformation is time-dependent. The deformation gradient may be decomposed as $\F=\Fe\Fv$, into an instantaneous distortion $\Fe$ and a time-dependent viscous distortion $\Fv$, both generally incompatible.
The local configuration reached by the instantaneous elastic unloading $\Fe^{-1}$ is stressed, in general, and the material manifold is the Euclidean manifold $(\mathcal{B},\Go)$ throughout \citep{SaYa2024viscoelasticity}.
The additive decomposition of the free energy into equilibrium and non-equilibrium parts forces the tensorial characters $\Fv(X):T_X\mathcal{B}\to T_X\mathcal{B}$ and $\Fe(X):T_X\mathcal{B}\to T_x\mathcal{S}$, and reads
%---------------------------------
\begin{equation} \label{eq:visco-energy}
	\mathring{\mathcal W}=\mathring{\mathcal W}_{ \mathrm{eq}}(X,\Theta,\F,\mathring{\boldsymbol\Lambda},\Go,\g)
	+\mathring{\mathcal W}_{ \mathrm{neq}}(X,\Theta,\Fe,\mathring{\boldsymbol\Lambda},\Go,\g)\,,
\end{equation}
%---------------------------------
where $\mathring{\boldsymbol{\Lambda}}$ denotes the structural tensors of the material symmetry group.
Applying the material-covariance argument of \eqref{Elasticity-Sym-Group-Reform}--\eqref{Eq:FreeAnelast-Reform} to the non-equilibrium branch with $\boldsymbol{\mathsf{T}}=\Fv$ yields the equivalent representation
%---------------------------------
\begin{equation} \label{eq:visco-transport}
	\mathring{\mathcal W}_{ \mathrm{neq}}(X,\Theta,\Fe,\mathring{\boldsymbol\Lambda},\Go,\g)
	=\mathcal W_{ \mathrm{neq}}(X,\Theta,\F,\Lambdav,\Gv,\g)\,,
	\quad
	\Gv=\Fv^*\Go\,,
	\quad
	\Lambdav=\Fv^*\mathring{\boldsymbol\Lambda}\,.
\end{equation}
%---------------------------------
Both branches thus depend on the total deformation gradient, yet they measure it against different reference metrics: $\Go$ in the equilibrium branch and $\Gv$ in the non-equilibrium one.
The kinematics designates neither metric; it is the constitutive representation that assigns a metric to the reference leg of $\F$, and the same motion presents a different geometric picture in each branch.

%---------------------------------
\emph{Equilibrium branch.}
The equilibrium picture is exactly that of elasticity: only the Euclidean metric is of relevance, the viscous distortion plays no part, and the geometric picture reduces to
%---------------------------------
\begin{equation} \label{eq:visco-chain-eq}
	(T_X\mathcal{B},\Go)\ \xrightarrow{\ \F\ }\ (T_x\mathcal{S},\g)\,.
\end{equation}
%---------------------------------
Hence, \S\ref{Subsection:Phys-Elasticity} applies without alteration.
Recalling that $\boldsymbol\ell$ and $\Lo$ orthonormalize the spatial and referential coordinate frames with respect to $\g$ and $\Go$, the physical components of $\F$ read%\footnote{The distortions $\Fe$ and $\Fv$ \do not have, independently from their composition $\F$, any physical significance on the equilibrium branch.}
%---------------------------------
\begin{equation} \label{eq:visco-phys-eq}
	\hat{\mathrm{F}}^a{}_A=(\ell^{-1})_b{}^a\,\mathrm{F}^b{}_B\,\cLo_A{}^B\,.
\end{equation}
%---------------------------------

%---------------------------------
\emph{Non-equilibrium branch.}
The non-equilibrium picture resembles that of anelasticity without being it. Formally, the parallel is complete: $\Gv=\Fv^*\Go$, exactly as $\G=\Fa^*\Go$---cf.~Fig.~\ref{Fig:distortions}---so that $\Fv$ is an isometry and the picture is the chain\footnote{What distinguishes the two settings is the standing of the convected metric, not its construction: in anelasticity, $\G$ measures natural lengths and $(\mathcal{B},\G)$ is the material manifold of the body, whereas here the local intermediate configuration $(\mathcal{B},\Gv)$ is stressed and does not yield natural lengths.}
%---------------------------------
\begin{equation} \label{eq:visco-chain-neq}
	(T_X\mathcal{B},\Gv)\ \xrightarrow{\ \Fv\ }\ (T_X\mathcal{B},\Go)
	\ \xrightarrow{\ \Fe\ }\ (T_x\mathcal{S},\g)\,.
\end{equation}
%---------------------------------
To avoid ambiguity, we denote with ``$\tilde{(.)}$'' all physical components taken in the non-equilibrium picture, reserving ``$\hat{(.)}$'' for those of the equilibrium one. Accordingly, with $\Lv$ orthonormalizing the referential coordinate frame with respect to $\Gv$, as $\Lo$ and $\Lt$ do with respect to $\Go$ and $\G$,
%---------------------------------
\begin{equation} \label{eq:visco-phys-neq}
	\tilde{\mathrm{F}}^a{}_A
		=(\ell^{-1})_b{}^a\,\mathrm{F}^b{}_B\,\cLv_A{}^B\,,
	\quad
	\tilde{\cFe}{}^a{}_A=(\ell^{-1})_b{}^a\,\cFe^b{}_B\,\cLo_A{}^B\,,
	\quad
	\tilde{\cFv}{}^A{}_B=(\cLo^{-1})_C{}^A\,\cFv^C{}_D\,\cLv_B{}^D\,,
	\quad
	\tilde{\mathrm{F}}^a{}_A=\tilde{\cFe}{}^a{}_B\,\tilde{\cFv}{}^B{}_A\,.
\end{equation}
%---------------------------------

The two representations of $\F$ differ only in the referential frame to which its reference leg is referred, and are hence related by
%---------------------------------
\begin{equation} \label{eq:visco-N}
	\tilde{\mathrm{F}}^a{}_A=\hat{\mathrm{F}}^a{}_B\,\mathsf{N}_A{}^B\,,
	\quad\text{with}\quad
	\mathsf{N}_A{}^B=(\cLo^{-1})_C{}^B\,\cLv_A{}^C\,,
\end{equation}
%---------------------------------
where $\mathsf{N}$ carries the $\Go$-orthonormal referential frame to the $\Gv$-orthonormal one.
Both $\hat{\F}$ and $\tilde{\F}$ are dimensionally consistent, dual-compatible, and unit-normalized, hence physical in the sense of \S\ref{Section:Physical-Components}, and they report on the long-term and the instantaneous response, respectively. 
The choice between these interpretations is a constitutive one and depends on the branch of the material response relative to which the components are defined.
%---------------------------------
\paragraph{Visco-anelasticity.} In visco-anelasticity, eigenstrain and time-dependent relaxation are both present, and the deformation gradient admits the decomposition
%---------------------------------
\begin{equation} \label{eq:va-split}
	\F=\Fe\Fv\Fa=\Fve\Fa\,,
	\quad
	\Fve=\Fe\Fv\,.
\end{equation}
%---------------------------------
Two local intermediate configurations emerge, one stress-free---reached by $\Fa$---and one stressed---reached by $\Fv\Fa$; the material manifold is $(\mathcal{B},\G)$, as in anelasticity \citep{SaYa2026visco}. The free energy again splits additively, with the equilibrium branch depending on $\Fve$ and the non-equilibrium one on $\Fe$,
%---------------------------------
\begin{equation} \label{eq:va-energy}
	\mathring{\mathcal W}
	=\mathring{\mathcal W}_{ \mathrm{eq}}(X,\Theta,\Fve,\mathring{\boldsymbol\Lambda},\Go,\g)
	+\mathring{\mathcal W}_{ \mathrm{neq}}(X,\Theta,\Fe,\mathring{\boldsymbol\Lambda},\Go,\g)\,.
\end{equation}
%---------------------------------
Applying material covariance to the two branches with $\boldsymbol{\mathsf{T}}=\Fa$ and $\boldsymbol{\mathsf{T}}=\Fv\Fa$, respectively, yields
%---------------------------------
\begin{equation} \label{eq:va-transport}
	\mathring{\mathcal W}=\mathcal W_{ \mathrm{eq}}(X,\Theta,\F,\boldsymbol\Lambda,\G,\g)
	+\mathcal W_{ \mathrm{neq}}(X,\Theta,\F,\Lambdai,\Gi,\g)\,,
	\quad
	\Gi=(\Fv\Fa)^*\Go=\Fa^*\Gv\,,
	\quad
	\Lambdai=(\Fv\Fa)^*\mathring{\boldsymbol\Lambda}\,.
\end{equation}
%---------------------------------
Similarly to viscoelasticity, two geometric pictures emerge: the equilibrium branch measures $\F$ against the material metric $\G$, the non-equilibrium branch against the instantaneous metric $\Gi$.

%---------------------------------
\emph{Equilibrium branch.} The equilibrium picture is exactly that of anelasticity, with $\Fve$ replacing $\Fe$, the geometric picture being the chain
%---------------------------------
\begin{equation} \label{eq:va-chain-eq}
	(T_X\mathcal{B},\G)\ \xrightarrow{\ \Fa\ }\ (T_X\mathcal{B},\Go)
	\ \xrightarrow{\ \Fve\ }\ (T_x\mathcal{S},\g)\,.
\end{equation}
%---------------------------------
Thus \S\ref{Subsection:Phys-Anelasticity} applies without alteration, and the expressions for physical components follow as%\footnote{The distortions $\Fe$ and $\Fv$ \do not have, independently from their composition $\Fve$, any physical significance on the equilibrium branch.}
%---------------------------------
\begin{equation} \label{eq:va-phys-eq}
	\hat{\mathrm{F}}^a{}_A=(\ell^{-1})_b{}^a\,\mathrm{F}^b{}_B\,\cL_A{}^B\,,
	\quad
	\hat{\cFve}{}^a{}_A =(\ell^{-1})_b{}^a\,\cFve^b{}_B\,\cLo_A{}^B\,,
	\quad
	\hat{\cFa}{}^A{}_B=(\cLo^{-1})_C{}^A\,\cFa^C{}_D\,\cL_B{}^D\,,
	\quad
	\hat{\mathrm{F}}^a{}_A=\hat{\cFve}{}^a{}_B\,\hat{\cFa}{}^B{}_A\,.
\end{equation}
%---------------------------------

%---------------------------------
\emph{Non-equilibrium branch.} The non-equilibrium branch features a similar construction with one further distortion in the chain
%---------------------------------
\begin{equation} \label{eq:va-chain-neq}
	(T_X\mathcal{B},\Gi)\ \xrightarrow{\ \Fa\ }\ (T_X\mathcal{B},\Gv)
	\ \xrightarrow{\ \Fv\ }\ (T_X\mathcal{B},\Go)
	\ \xrightarrow{\ \Fe\ }\ (T_x\mathcal{S},\g)\,,
\end{equation}
%---------------------------------
since $\Fa^*\Gv=\Gi$ and $\Fv^*\Go=\Gv$.
With $\Li$ orthonormalizing the referential coordinate frame with respect to $\Gi$, the corresponding physical component representations are given by
%---------------------------------
\begin{equation} \label{eq:va-phys-neq}
\begin{gathered}
	\tilde{\mathrm{F}}^a{}_A=(\ell^{-1})_b{}^a\,\mathrm{F}^b{}_B\,\cLi_A{}^B\,,
	\quad
	\tilde{\cFe}{}^a{}_A=(\ell^{-1})_b{}^a\,\cFe^b{}_B\,\cLo_A{}^B\,,
	\quad
	\tilde{\cFv}{}^A{}_B=(\cLo^{-1})_C{}^A\,\cFv^C{}_D\,\cLv_B{}^D\,,
	\\
	\tilde{\cFa}{}^A{}_B=(\cLv^{-1})_C{}^A\,\cFa^C{}_D\,\cLi_B{}^D\,,
	\quad
	\tilde{\mathrm{F}}^a{}_A =\tilde{\cFe}{}^a{}_B\,\tilde{\cFv}{}^B{}_C\,\tilde{\cFa}{}^C{}_A\,.
\end{gathered}
\end{equation}
%---------------------------------

The two representations are related by
%---------------------------------
\begin{equation} \label{eq:va-N}
	\tilde{\mathrm{F}}^a{}_A=\hat{\mathrm{F}}^a{}_B\,\mathsf{N}_A{}^B\,,
	\quad\text{with}\quad
	\mathsf{N}_A{}^B=(\cL^{-1})_C{}^B\,\cLi_A{}^C\,.
\end{equation}
%---------------------------------
%---------------------------------
\section{Example: A circular cylindrical bar with distributed eigentwists}
\label{Section:Example}
%-----------------------------
\paragraph{Geometry and kinematics.}
Consider a solid circular cylindrical bar of radius $R_o$, and let the
ambient and reference spaces be endowed with cylindrical coordinates $(r,\theta,z)$ and
$(R,\Theta,Z)$, respectively. The spatial Euclidean metric $\g$ and its inherited
referential copy $\Go=\iota^*\g$ have the diagonal matrix representations $\g =\operatorname{diag}\left\{1,\,r^2,\,1\right\}$, and $\Go =\operatorname{diag}\left\{1,\,R^2,\,1\right\}$ such that the axis of the bar in its reference state is aligned with $Z$.
Let the bar be endowed with an axisymmetric distribution of torsional eigenstrains of
eigentwist density $\psi(R)$---note that $[\psi]=L^{-1}$. Cartan's moving frames yield the associated anelastic
distortion $\Fa$ \citep{YavariGoriely2015Twist-Fit}, whence the material metric
$\G=\Fa^*\Go$; the two have the representations
%-----------------------------
\begin{equation}
	\Fa =\begin{bmatrix}
	1&0&0\\
	0&1&\psi(R)\\
	0&0&1
	\end{bmatrix}\,,
	\quad
	\G =
	\begin{bmatrix}
		1 & 0 & 0 \\
		0 & R^2 & R^2 \psi(R) \\
		0 & R^2 \psi(R) & 1 + R^2 \psi^2(R)
	\end{bmatrix}\,.
\end{equation}
%----------------------------------
In general, when $\psi'(R)\neq0$, the material metric $\G$ is not induced by a deformation of a Euclidean reference configuration and hence represents an incompatible distribution of torsional eigenstrains.

The cylindrical coordinates $(r,\theta,z)$ and $(R,\Theta,Z)$ are orthogonal with respect
to $\g$ and $\Go$, respectively, so that orthonormalization on these legs is effected, as in \eqref{eq:Ortho_Coord_M}, by the following diagonal maps
%-----------------------------
\begin{equation}
	\left\llbracket\cLo_A{}^B\right\rrbracket
	=\operatorname{diag}\left\{1,\frac{1}{R},\,1\right\}\,,
	\quad
	\left\llbracket\ell_a{}^b\right\rrbracket
	=\operatorname{diag}\left\{1,\,\frac{1}{r},\,1\right\}\,.
\end{equation}
%-----------------------------
With respect to $\G$, however, $(R,\Theta,Z)$ is not orthogonal, and a $\G$-orthonormal
moving coframe field and its dual are given by
%-----------------------------
\begin{subequations}
\label{Cartan-Moving-Twist}
\begin{align}
	\hat{\boldsymbol\Omega}^1 &= dR\,,
	&\hat{\boldsymbol\Omega}^2 &= R\,d\Theta+R\psi(R)\,dZ\,,
	&\hat{\boldsymbol\Omega}^3 &= dZ\,,
	\label{OrthonormalFrameTwist}\\
	\hat{\mathbf{E}}_1 &= \frac{\partial}{\partial R}\,,
	&\hat{\mathbf{E}}_2 &= \frac{1}{R}\frac{\partial}{\partial\Theta}\,,
	&\hat{\mathbf{E}}_3 &= \frac{\partial}{\partial Z}-\psi(R)\frac{\partial}{\partial\Theta}\,.
	\label{Moving-Frame-Twist}
\end{align}
\end{subequations}
%-----------------------------
Recast as in \eqref{eq:2pt-frames}--\eqref{eq:2pt-coframes}, i.e. $\hat{\boldsymbol\Omega}^A=(\cL^{-1})_B{}^A\,dX^B$
and $\hat{\mathbf E}_A=\cL_A{}^B\,\partial_B$, these read
%-----------------------------
\begin{equation}
	\begin{bmatrix}
		\hat{\boldsymbol\Omega}^1 &
		\hat{\boldsymbol\Omega}^2 &
		\hat{\boldsymbol\Omega}^3
	\end{bmatrix}
	=
	\begin{bmatrix}
		dR & d\Theta & dZ
	\end{bmatrix}
	\begin{bmatrix}
		1 & 0 & 0\\
		0 & R & 0\\
		0 & R\,\psi(R) & 1
	\end{bmatrix}\,,
	\quad
	\begin{bmatrix}
		\hat{\mathbf{E}}_1\\
		\hat{\mathbf{E}}_2\\
		\hat{\mathbf{E}}_3
	\end{bmatrix}
	=\begin{bmatrix}
		1 & 0 & 0\\
		0 & \frac{1}{R} & 0\\
		0 & -\psi(R) & 1
	\end{bmatrix}
	\begin{bmatrix}
		\partial_R\\
		\partial_\Theta\\
		\partial_Z
	\end{bmatrix}\,,
\end{equation}
%-----------------------------
with the orthonormalization transformation $\Lt$ given by
%-----------------------------
\begin{equation}
	\left\llbracket\cL_A{}^B\right\rrbracket
	=\begin{bmatrix}
		1 & 0 & 0\\
		0 & \frac{1}{R} & 0\\
		0 & -\psi(R) & 1
	\end{bmatrix}\,,
	\quad\text{or equivalently}\quad
	\left\llbracket\big(\cL^{-1}\big)_B{}^A\right\rrbracket
	= \begin{bmatrix}
		1 & 0 & 0\\
		0 & R & 0\\
		0 & R\,\psi(R) & 1
	\end{bmatrix}\,.
\end{equation}
%-----------------------------

Assuming an isotropic material, one may take the following deformation ansatz
${r=r(R)}$, ${\theta = \Theta + \gamma(Z)}$, ${z=\lambda^2 Z}$
with azimuthal shear $\gamma=\gamma(Z)$ and axial stretch $\lambda^2$---both dimensionless. The deformation gradient and the elastic distortion may hence be given by the following representations
%-----------------------------
\begin{equation}
	\mathbf{F}=
	\begin{bmatrix}
		r'(R) & 0 & 0\\
		0 & 1 & \gamma'(Z)\\
		0 & 0 & \lambda^2
	\end{bmatrix}\,.
\qquad
	\Fe=
	\begin{bmatrix}
		r'(R) & 0 & 0 \\[1mm]
		0 & 1 & \gamma'(Z)-\psi(R) \\[2mm]
		0 & 0 & \lambda^2
	\end{bmatrix}\,.
\end{equation}
%-----------------------------

%-----------------------------
\paragraph{Physical components.}
This example realizes a case of practical interest of \S\ref{Subsection:Phys-Anelasticity}: the referential chart is adapted to $\Go$ but not to $\G$. Accordingly, $\Fe$ reduces on both legs by \eqref{eq:phys-orth-anel}, whereas the referential leg of $\F$ retains the non-diagonal $\Lt$, and $\Fa$ takes $\Lt$ on its domain and $\Lo$ on its codomain. From \eqref{eq:Fe-phys-general}--\eqref{eq:F-phys-general} one obtains
%-----------------------------
\begin{equation}
\label{eq:Twist_Physi_eigen}
	\hat{\F}=\hat{\Fe}=
	\begin{bmatrix}
		r'(R) & 0 & 0 \\[1mm]
		0 & \dfrac{r(R)}{R} & r(R)\left(\gamma'(Z)-\psi(R)\right) \\[2mm]
		0 & 0 & \lambda^2
	\end{bmatrix}\,,
	\quad
	\hat{\Fa}=\llbracket\delta^A_B\rrbracket\,.
\end{equation}
%-----------------------------
The anelastic distortion is thus trivialized by Cartan's construction \eqref{Cartan-Moving-Twist}: it delivers the $\G$-orthonormal frame given by the $\Fa$-image of $\Go$---precisely the gauge of Remark~\ref{remark:split-consistency}---and yields $\hat{\F}=\hat{\Fe}$ by \eqref{eq:Fhat-composition}. The two representations coincide under the chosen physical component orthonormalization while $\F\neq\Fe$ as two-point tensors, their referential legs being attached to different configurations. Recall that the anelastic content is not thereby lost: it enters $\hat{\F}$ through $\Lt$, and hence through $\G=\Fa^*\Go$.
The mismatch between the realized twist and the prescribed eigentwist appears in the shear entry alone, while the bulk gives $\det\llbracket\hat{\cF}^a{}_A\rrbracket=\lambda^2\,r(R)\,r'(R)/R=J$ in accordance with Proposition~\ref{prop:Jacobian}. The coordinate components carry mixed physical dimensions, $[\cF^\theta{}_Z]=[\gamma']=L^{-1}$ while $\cF^r{}_R$ is dimensionless, whereas every entry of $\hat{\F}$ is dimensionless.
%-----------------------------
%---------------------------------
\begin{remark}[Comparison with Aris's construction]
\label{remark:Aris-twist}
As discussed in Remark~\ref{Remark:Aris}, \citet{Aris1962} normalizes the coordinate basis vectorsbut does not orthogonalize them. Here the referential chart is $\G$-non-orthogonal, so that $\hat{\mathbf E}_{(A)}=\partial_A/\sqrt{\cG_{AA}}$ is of unit length but non-orthogonal. Writing $S(R)=\sqrt{\cG_{ZZ}(R)}=\sqrt{1+R^2\psi^2(R)}$, the resulting components read
%---------------------------------
\begin{subequations}
\begin{alignat}{2}
	\label{eq:Twist_F_Aris}
	\llbracket\hat{\cF}^a{}_{(A)}\rrbracket
	&=\begin{bmatrix}
		r'(R) & 0 & 0\\[1mm]
		0 & r(R)/R & r(R)\,\gamma'(Z)/S(R)\\[1mm]
		0 & 0 & \lambda^2/S(R)
	\end{bmatrix}\,, \quad&
	\det\llbracket\hat{\cF}^a{}_{(A)}\rrbracket&=J/S(R)\,,\\
	\label{eq:Twist_Fe_Aris}
	\llbracket\hat{\cFe}^a{}_{(A)}\rrbracket
	&=\begin{bmatrix}
		r'(R) & 0 & 0\\[1mm]
		0 & r(R)/R & r(R)\left(\gamma'(Z)-\psi(R)\right)\\[1mm]
		0 & 0 & \lambda^2
	\end{bmatrix}\,, \quad&
	\det\llbracket\hat{\cFe}^a{}_{(A)}\rrbracket&=J\,,\\
	\label{eq:Twist_Fa_Aris}
	\llbracket\hat{\cFa}^A{}_{(B)}\rrbracket
	&=\begin{bmatrix}
		1 & 0 & 0\\[1mm]
		0 & 1 & R\,\psi(R)/S(R)\\[1mm]
		0 & 0 & 1/S(R)
	\end{bmatrix}\,, \quad&
	\det\llbracket\hat{\cFa}^A{}_{(B)}\rrbracket&=1/S(R)\,.
\end{alignat}
\end{subequations}
%---------------------------------
In this example, Aris's normalization differs from the framework advanced in this work only on the leg measured against $\G$, the remaining legs being referred to the already orthogonal (and Euclidean) $\Go$ and $\g$. Of the three distortions, only $\Fe$ carries no material leg, and Aris's construction accordingly reproduces \eqref{eq:Twist_Physi_eigen} for it alone---since, otherwise, $\F$ is sitting on $\G$ and $\g$ and $\Fa$ on $\G$ and $\Go$. The composition $\hat{\Fe}\,\hat{\Fa}=\hat{\F}$ is nonetheless preserved, as the factors adopted on the material leg cancel between the two distortions.
The Jacobian is recovered as the determinant of $\hat{\Fe}$ but not as that of $\hat{\F}$, contrary to Proposition~\ref{prop:Jacobian}.
Further, $\llbracket\hat{\cFa}^A{}_{(B)}\rrbracket$ is not orthogonal, although $\Fa$ is an isometry from $(\mathcal{B},\G)$ onto $(\mathcal{B},\Go)$.
Finally, we observe that the twist mismatch $\gamma'(Z)-\psi(R)$ no longer features in \eqref{eq:Twist_F_Aris}, the off-diagonal entry reading $r(R)\gamma'(Z)/\sqrt{1+R^2\psi^2(R)}$, with the applied azimuthal shear in the numerator and the eigentwist only in the normalization factor $\sqrt{\cG_{ZZ}(R)}=\sqrt{1+R^2\psi^2(R)}$. Yet it is this very mismatch that sources the residual stress in this problem.
\end{remark}
%---------------------------------

%---------------------------
\section{Conclusions}\label{Section:Conclusions}
Physical components have long been used in continuum mechanics and physics as quantities that are presumed to correspond more directly to physical measurements than coordinate components. In orthogonal coordinate systems, the classical construction is unambiguous: normalization of the coordinate basis yields an orthonormal frame, and physical components are obtained by expressing tensor fields relative to this frame and its dual coframe. In non-orthogonal coordinates, however, normalization of the individual coordinate directions does not produce an orthonormal frame, and no analogous canonical construction exists.

In this paper, physical components are formulated using orthonormal frames on a Riemannian manifold. Starting from an arbitrary coordinate basis, an orthonormal frame may be constructed by orthonormalization, and physical components are defined as the components of tensor fields relative to this frame and its dual coframe. The orthonormal frame is not unique: different choices are related by local orthogonal transformations, and the corresponding physical components therefore possess an orthogonal gauge freedom. This framework recovers the classical formulae in orthogonal coordinates while remaining applicable in arbitrary coordinate systems.
The construction also extends to two-point tensors, for which the domain and codomain legs must be orthonormalized with respect to their corresponding metrics and admit independent orthogonal gauge transformations.

A general normalization framework is introduced in order to identify the precise assumptions underlying the orthonormal-frame construction. Three requirements are distinguished: dimensional consistency of the normalized components, dual compatibility of the frame and coframe, and unit normalization of both. Dimensional consistency alone leaves independent dimensionless general linear gauge freedoms for the contravariant and covariant indices. Dual compatibility locks these two freedoms into a single general linear gauge freedom by requiring the normalized frame and coframe to arise as a dual pair. It was then shown that a dual-compatible frame and coframe can both have unit length if and only if they are orthonormal. Consequently, imposing all three requirements reduces the general linear ambiguity precisely to the local orthogonal gauge freedom of the orthonormal-frame construction.

This characterization shows that the orthogonal gauge freedom is not an artefact of a particular orthonormalization procedure. Rather, it is the residual freedom left undetermined by dimensional consistency, dual compatibility, and unit normalization. 
Thus, up to a local orthogonal transformation, the orthonormal-frame construction is the only one that yields components with the same physical dimensions as the tensor field while preserving dual compatibility between contravariant and covariant representations. It is also the only construction in which both the frame and its dual coframe are unit normalized.
For two-point tensors, the corresponding freedom consists of independent orthogonal transformations of the orthonormal frames associated with the domain and codomain metrics.

Our analysis also clarifies the status of alternative constructions proposed in the literature. In particular, \citet{Aris1962} defined physical components relative to individually normalized coordinate directions rather than an orthonormal frame. In non-orthogonal coordinates, the frame and coframe generated by this construction are related by the metric-dual rather than by natural duality. 
The construction therefore satisfies dimensional consistency and unit normalization but not dual compatibility.
Its dual and mirror variants exhibit the same obstruction: each satisfies two of the three requirements, but none satisfies all three unless the coordinate directions are orthogonal. Such constructions are generally not related to the orthonormal-frame construction by orthogonal transformations. Hence, if physical components are understood merely as dimensionally normalized components, their non-uniqueness is broader than orthonormal-frame freedom and carries a general linear gauge character.

The framework developed in this paper is specialized to nonlinear elasticity and to three classes of inelastic response: finite anelasticity, finite viscoelasticity, and finite visco-anelasticity.% in which anelastic distortion and viscous relaxation are both present \citep{SaYa2026visco}.
Physical components of the deformation gradient, inelastic distortions, strain measures, and stress tensors are obtained by expressing each tensorial leg relative to an orthonormal frame associated with the metric on its corresponding space. This distinction is particularly important in beyond elasticity, where the material, Euclidean reference, and spatial metrics need not coincide. The example involving distributed eigentwists illustrates explicitly how the resulting physical components differ from coordinate components and how the relevant metric enters each construction.

These results clarify earlier observations concerning the ambiguity of physical components in non-orthogonal coordinates and place them within a broader geometric framework. Physical components are useful, dimensionally consistent, and normalized representations of tensorial quantities, but they are not intrinsic objects. In general, one should speak of \emph{a} set of physical components associated with specified metrics and chosen frames, rather than \emph{the} physical components of a tensor field.

%%---------------------------
%%------------------------------
%\section*{Acknowledgement}

%-------------------------
%----------------------------
\bibliographystyle{abbrvnat}
\bibliography{Phys_Comp}
\end{document}